\documentclass[11pt]{article}
\usepackage[activeacute,english]{babel} 
\usepackage[utf8]{inputenc}
\usepackage{verbatim} 
\usepackage{graphicx} 
\usepackage[usenames,dvipsnames]{color} 
\usepackage{enumerate,mdwlist}  
\usepackage{multirow} 
\usepackage{amsmath,amsfonts,amssymb,amsthm} 
\usepackage{bm} 
\usepackage{empheq} 
\usepackage{bbm,dsfont} 
\usepackage{mathrsfs} 
\numberwithin{equation}{section} 
\usepackage{tikz-cd,tkz-graph,tkz-arith,tkz-berge} 
\usepackage{tkz-graph,tkz-arith,tkz-berge} 
\usetikzlibrary{patterns,calc}
\usepackage[small,nohug,heads=littlevee]{diagrams}
\diagramstyle[labelstyle=\scriptstyle]
\usepackage{float} 
\usepackage[labelfont=bf,labelsep=space]{caption} 
\usepackage{empheq} 

\usepackage[OT2,OT1]{fontenc}
\DeclareRobustCommand\cyr{%
  \renewcommand\rmdefault{wncyr}%
  \renewcommand\sfdefault{wncyss}%
  \renewcommand\encodingdefault{OT2}%
  \normalfont
  \selectfont}
\DeclareTextFontCommand{\textcyr}{\cyr}

\usepackage{subcaption}

\usepackage{geometry}
\usepackage{color}
\definecolor{red}{rgb}{.7,0,0}
\definecolor{blue}{rgb}{0,0,1}

\def\mcG{\mathcal{G}}
\def\mcH{\mathcal{H}}

\def\mcM{\mathcal{M}}

\def\mcP{\mathcal{P}}

\def\mcS{\mathcal{S}}
\def\mcO{\mathcal{O}}
\def\mcX{\mathcal{X}}
\def\mcU{\mathcal{U}}

\def\bbR{\mathbb{R}}

\def\bbN{\mathbb{N}}

\def\bbS{\mathbb{S}}
\def\fkC{\mathfrak{C}}

\def\sgn{\mathsf{sgn}}

\def\bfdelta{\boldsymbol{\delta}}
\def\bfeps{\boldsymbol{\varepsilon}}
\def\bfI{\boldsymbol{I}}
\def\rew{\rightsquigarrow}
\def\rews{\mathrel{\raisebox{-2.5pt}{$\rew$}}}
\newcommand{\subs}[1]{\overset{#1}{\rews}}

\def\syd{\mathcal{D}} 
\def\sye{\mathcal{E}} 
\def\syc{\mathfrak{c}} 

\def\Izero{\mathbin{\vcenter{\hbox{$\bullet$}}}}
\def\Ione{\between}

\usepackage{ifpdf}
\ifpdf
\usepackage[pdfencoding=auto,unicode]{hyperref} 
\hypersetup{
 pdftoolbar=true,
 pdfmenubar=true,
 pdfstartview={FitV},
 pdftitle={Computing the Homology of Semialgebraic Sets. I: Lax Formulas},
 pdfauthor={Peter B\"{u}rgisser, Felipe Cucker, Josu\'{e} Tonelli-Cueto},
 pdfsubject={semialgebraic geometry, algebraic topology, numerical algorithms},
 pdfkeywords={semialgebraic geometry,algebraic topology},
 pdfcreator={overleaf}, 
 pdfproducer={overleaf},
 linkcolor=red     
 citecolor=green   
 urlcolor=cyan     
 filecolor=magenta 
}
\fi
\usepackage{bookmark}

\title{Computing the Homology of Semialgebraic Sets.\\
II: General formulas\thanks{This work was supported by the Einstein Foundation, Berlin.}}
\author{Peter B\"urgisser\thanks{Partially funded by the European Research Council (ERC) under the European's Horizon 2020 research and innovation programme (grant agreement No 787840).}
\\
Technische Universit\"at Berlin\\ 
Institut f\"ur Mathematik\\ 
GERMANY\\
{\tt pbuerg@math.tu-berlin.de} 
\and
Felipe Cucker\thanks{Partially supported by a GRF grant
from the Research Grants Council of the Hong Kong
SAR (project number CityU 11302418).}
\\
Dept. of Mathematics\\
City University of Hong Kong\\
HONG KONG\\
{\tt macucker@cityu.edu.hk}
\and
Josu\'{e} Tonelli-Cueto\thanks{Partially supported by ANR JCJC
GALOP (ANR-17-CE40-0009), the PGMO grant ALMA, and the PHC GRAPE.}
\\
Inria Paris \& IMJ-PRG\\
OURAGAN team\\
Sorbonne Universit\'e\\
Paris, FRANCE\\
{\tt  josue.tonelli.cueto@bizkaia.eu}
}

\makeatletter
\def\th@plain{%
  \thm@notefont{}
  \slshape 
}
\def\th@definition{%
  \thm@notefont{}
  \normalfont 
}
\makeatother
\theoremstyle{plain}
\newtheorem{lem}{Lemma}[section]
\newtheorem{prop}[lem]{Proposition}
\newtheorem{theo}[lem]{Theorem}

\theoremstyle{definition}
\newtheorem{defi}[lem]{Definition}
\theoremstyle{remark}
\newtheorem{question}[lem]{Question}

\newtheorem{exam}[lem]{Example}

\newtheorem{remark}[lem]{Remark}
\newcommand{\eproof}{\hfill\qed}

\newcommand{\Ap}{\mathsf{S}}
\newcommand{\cost}{\mathsf{cost}}
\newcommand{\size}{\mathsf{size}}
\def\bfd{\boldsymbol{d}}
\def\Hd{\mcH_{\bfd}}
\def\Pd{\mcP_{\bfd}}
\def\hm{^{\mathsf{h}}}
\def\kappabar{\overline{\kappa}}
\def\kappaff{\overline{\kappa}_{\sf aff}}

\newcommand{\cech}[2]{\mathrm{\check{C}}_{#1}\big(#2\big)}
\def\Oh{\mathcal{O}}
\def\gv{\textsf{\cyr G\!V}}
\def\Pe{\textrm{\cyr P}}
\def\pe{\textrm{\cyr p}}
\def\Tg{\mathrm{T}}
\def\sfH{\mathop{\mathsf H}}
\def\sfK{\mathsf K}
\def\Hm{\sfH}
\def\Oh{\mathcal{O}}
\def\diff{\mathrm{D}}
\def\diffa{\overline{\diff}}

\def\affgamma{\overline{\gamma}}

\def\Phibar{\overline{\Phi}}
\def\fbar{\overline{f}}
\def\fkC{\mathfrak C}
\def\fkA{\mathfrak A}
\def\bflambda{\boldsymbol{\lambda}}

\usepackage{multicol}
\usepackage[ruled,algosection,vlined]{algorithm2e}
\SetKwInOut{postcondition}{Postcondition}
\SetKwInOut{precondition}{Precondition}

\usepackage{rotating,lipsum}

\begin{document}
\date{}
\maketitle

\begin{abstract}
We describe and analyze a numerical algorithm for computing the homology (Betti numbers and torsion coefficients) of semialgebraic sets given by Boolean formulas. 
The algorithm works in weak exponential time. This means that 
outside a subset of data having exponentially small measure, the 
cost of the algorithm is single exponential in the size of the data. This extends the work in~\cite{BCTC1} 
to {\em arbitrary} semialgebraic sets. 

All previous algorithms proposed for this problem have doubly 
exponential complexity.
\end{abstract}
\section{Introduction}

This paper is a continuation of~\cite{BCTC1}. In the latter, we exhibited a numerical 
algorithm computing the topology of a {\em closed} semialgebraic set described by a
{\em monotone} Boolean combination of polynomial equalities and {\em lax} inequalities. This 
restriction, that the formula defining the Boolean combination is lax, forces 
connected components of the (projective closure of the) semialgebraic set to 
be separated by a positive distance. 
The fact that we can control these distances by the condition number of the 
tuple of polynomials in the description of the set was central 
in the design of the algorithm. 

Our goal in this paper is to exhibit a numerical algorithm, with similar complexity bounds 
to that in~\cite{BCTC1}, but working for {\em arbitrary} semialgebraic sets. That is, 
complements are allowed in the Boolean combinations (or, equivalently, negations 
in the Boolean formulas) describing semialgebraic sets and so are strict 
inequalities. This algorithm works in {\em weak exponential time}. This means that 
its cost, or running time, is single exponential outside an exceptional set whose 
measure vanishes exponentially fast. Outside this exceptional set then, our 
algorithm works exponentially faster than state-of-the-art algorithms (which are 
doubly exponential). 
For a background on numerical algorithms, 
condition, and weak cost ---and to avoid being repetitious--- we refer the reader 
to the introduction of~\cite{BCTC1}. We thus proceed with a description of the 
problem and our main result. 

In all what follows we fix natural numbers $q\geq 1,n\ge 2$ 
and a tuple 
$\bfd:=(d_1,\ldots,d_q)\in\bbN^q$. We denote by $\Pd[q]$ the linear space of 
polynomial tuples $p=(p_1,\ldots,p_q)$ with $p_i\in\bbR[X_1,\ldots,X_n]$ 
of degree at most $d_i$. 
We let $D:=\max\{d_1,\ldots,d_q\}$ 
and $N:=\dim \Pd[q]$. We consider the latter to be the size of the elements 
in $\Pd[q]$ as this is the number of real numbers we pass to an algorithm 
to specify a tuple $p\in\Pd[q]$. 

An element $p\in\Pd[q]$ determines $5q$ {\em atomic sets} of the form 
$\{x\in\bbR^n\mid p_i\propto 0\}$, where $\propto\in\{<,\leq,=,\ge,>\}$ 
and $i\in[q]:=\{1,\ldots,q\}$. A Boolean combination (Boolean formula) of these sets 
is an expression recursively constructed from them by taking unions, intersections,  
and complements (respectively,  disjunctions $\vee$, conjunctions $\wedge$, 
and negations $\neg$). We will often refer to a {\em Boolean formula 
over $p$} to mean a formula in the $5q$ atomic relations above. For such a 
formula $\Phi$ we denote by $\size(\Phi)$ the number of terms in its recursive 
construction. The quantity $\size(p,\Phi):=N+\size(\Phi)$ is therefore a measure 
of the size of a pair $(p,\Phi)$. 

Any such pair defines a semialgebraic set $W(p,\Phi)\subseteq\bbR^n$ and we 
are interested here in computing, with input $(p,\Phi)$, 
the homology groups of $W(p,\Phi)$. We already observed that $\size(p,\Phi)$ 
measures the input size for such an algorithm. Its {\em cost} on this input 
is the number of arithmetic operations and comparisons in $\bbR$ performed 
during the computation. In our algorithm, as in many numerical algorithms, 
this cost depends on $\size(p,\Phi)$ but it is not bounded by a function of this 
quantity only. Instead, the cost (and the precision required to ensure a correct 
output when running the algorithm with finite precision) depends as well on 
a condition number $\kappaff(p)$. This is a number in $[1,\infty]$. Tuples 
$p\in\Pd[q]$ with $\kappaff(p)=\infty$ are those for which, for some 
$\Phi$, arbitrarily small perturbations of the coefficients in $p$ may change 
the homology groups of $W(p,\Phi)$. Such tuples, called {\em ill-posed}, form a 
lower-dimensional semialgebraic subset of $\Pd[q]$. 
We briefly recall some facts about $\kappaff(p)$,  
and point to the relevant sections in~\cite{BCTC1} 
where these facts are shown, in~\S\ref{sec:kappa} below.

To obtain weak complexity bounds we need to endow $\Pd[q]$ with a probability measure and 
to do so, it will be convenient to have an inner product on this space. We 
endow $\Pd[q]$ with the Weyl inner product (see~\cite[\S3.1]{BCTC1}) and 
the associated unit sphere $\bbS(\Pd[q])=\bbS^{N-1}$ with the uniform 
probability measure. With this measure at hand we can state our main 
result (whose general form is remarkably similar to the one of~\cite{BCTC1}, 
the only difference being that the result here applies to {\em arbitrary} Boolean 
formulas). 

\begin{theo}\label{thm:main_result}
We exhibit a stable numerical algorithm~{\sf Homology} 
that, given a tuple $p\in\Pd[q]$ and a Boolean 
formula $\Phi$ over $p$, computes the homology groups of~$W(p,\Phi)$. 
The cost of~{\sf Homology} on input~$(p,\Phi)$,
denoted~$\cost(p,\Phi)$, satisfies:
\begin{description}
\item[(i)] $\cost(p,\Phi)\leq 
\size(\Phi) q^{\Oh(n)} (nD\kappaff(p))^{\Oh(n^2)}$.
\end{description}
Furthermore, if $p$ is drawn from the uniform distribution on $\bbS^{N-1}$, then:
\begin{description}
\item[(ii)] $\cost(p,\Phi)\leq \size(\Phi) q^{\Oh(n)}
(nD)^{\Oh(n^3)}$ with probability at least 
$1-(nq D)^{-n}$, and 
\item[(iii)]
$\cost(p,\Phi)\leq 
2^{\mcO\big(\size(p,\Phi)^{1+\frac{2}{D}}\big)}$ 
with probability at least $1-2^{-\size(p,\Phi)}$.
\end{description}
\end{theo}

A few comments on~Theorem~\ref{thm:main_result}:
\begin{description}
\item[(i)]
As in~\cite{BCTC1}, we direct the reader to Section~7 of~\cite{CKS16} for 
an explanation, along with a proof, of the numerical stability
mentioned in the statement above. Details can also be found in ~\cite[4\textsuperscript{\S3}-2]{tonellicuetothesis}.
\item[(ii)]
Part~(iii) of Theorem~\ref{thm:main_result} shows that {\sf Homology} 
works in {\em weak exponential time}. 
\item[(iii)]
It is easy to check that all the routines in algorithm {\sf Homology} do parallelize for the computation of the Betti numbers\footnote{Claims in~\cite{BCL17,BCTC1} regarding the parallelization of the computation of torsion coefficients were inaccurate. As of today, the main difficulty lies in showing the existence of efficient parallel algorithms for the computation of the Smith Normal Form of integer matrices. See~\cite[p.~160-161]{tonellicuetothesis} for more details.}.
The parallel version of the algorithm can then be shown to work in parallel time 
$\size(p,\Phi)^{\mcO(1)}$ with probability at least 
$1-2^{-\size(p,\Phi)}$. 
That is, it works in {\em weak parallel polynomial time}. The arguments 
for this are in~\cite[\S7.4]{BCTC1} (see also~\cite[4\textsuperscript{\S3}-1]{tonellicuetothesis}). 
\item[(iv)] We note that $\Phi$ can be rewritten as a formula in disjunctive normal form of size at most $(qD)^{\Oh(n)}$. Therefore one could in principle eliminate $\size(\Phi)$ from the complexity estimates. However, we found it desirable to indicate the dependence of the cost of the algorithm in terms of all the intervening parameters, including the size of the formula.
\end{description}
\medskip

\medskip

\noindent
{\bf Structure of the paper.}\quad 
In Section~\ref{sec:overview} we provide an overview of the 
various ingredients that make up our algorithms and its 
analysis. At the end of this section we are in a situation 
of describing the algorithm itself and give a proof of 
Theorem~\ref{thm:main_result} based on the notions and 
results in this overview. Sections~\ref{sec:prrofGV} 
and~\ref{sec:gamma} provide the proofs of these results. 
Finally, we conclude in Section~\ref{sec:hybrid} with a discussion on how to combine numeric with symbolic algorithms.
\medskip

{\small
\tableofcontents
}
\medskip

\noindent
{\bf Acknowledgments.}\quad
We are grateful to Nicolai Vorobjov who pointed us to (what we call here) Gabrielov-Vorobjov approximations. 

\section{Overview of the Algorithm}\label{sec:overview}

\subsection{Elimination of negations and lax inequalities}

The initial step in our algorithm eliminates negations and lax inequalities 
(in this order) in the given formula $\Phi$. This can be done in time linear 
in $\size(\Phi)$ and the resulting formula after this elimination 
has a size which is at most $2\,\size(\Phi)$, the increase being due to the 
substitutions 
$$
    p\geq 0 \rew (p=0 \vee p>0)
$$
and 
$$
    p\leq 0 \rew (p=0 \vee p<0).
$$
All along this paper, we will write $\alpha \rew\beta$
to indicate that an expression $\alpha$ is rewritten as (i.e., replaced by) 
another expression $\beta$. 

The resulting formula is therefore monotone (no negations) and 
has no lax inequalities: it is built over the $3q$ atoms 
$p_i\propto 0$ with $i\in[q]$ and $\propto\in\{<,=,>\}$. 
We will call these formulas {\em strict}. 

Strict formulas can be rewriten in Disjunctive Normal Form (DNF) and the 
resulting formula is also strict (as atoms remain unchanged). Even though 
we will not need to convert the input formula into DNF in our algorithm, a conversion 
that may exponentially increase its size, we will use DNFs 
in many of our reasonings. We therefore recall that we call {\em purely 
conjunctive} a conjunction of atoms and that semialgebraic sets given by 
purely conjunctive formulas are called {\em basic}.

\subsection{Homogeneization}

From now on, we assume that $\Phi$ is strict. The next step in our algorithm 
maps into the sphere the semialgebraic set $W(p,\Phi)$ by considering its 
spherical closure. To do so, we recall some notation. 

As before, let $\bfd =(d_1,\ldots,d_q)$ be a $q$-tuple of positive integers. 
We denote by $\Hd[q]$ the 
vector space of $q$-tuples $f=(f_1,\ldots,f_q)$ of  homogeneous polynomials 
$f_i\in\bbR[X_0,\ldots,X_n]$ of degree~$d_i$. 
We put $\bfd^* := (1,\bfd)$.
The {\em homogenization map} $\Hm:\Pd[q]\rightarrow 
\mcH_{\bfd^*}[q+1]$ is defined by
\begin{equation*}
  p\mapsto \Hm(p):=(\|p\|X_0,p_1\hm,\ldots,p_q\hm),
\end{equation*}
where $p_i\hm:=p_i\left(X_1/X_0,\ldots,X_n/X_0\right)X_0^{d_i}$ denotes 
the homogenization of $p_i$ and $\|p\|$ stands for the Weyl norm of 
the tuple~$p$ (see~\cite[\S3.1]{BCTC1}). 

Any formula $\Phi$ over $f\in\Hd[q]$ 
defines a semialgebraic subset $\Ap(f,\Phi)$ of the sphere~$\bbS^n$. 
It will be convenient to call these sets {\em spherical semialgebraic}. 
In order to simplify the notation, we will also write $\Ap(f=0)$ etc.\ with 
the obvious meaning. 

The following result is straightforward. 

\begin{prop}\label{generaltospherical}
Let $p\in\Pd[q]$ 
and $\Phi$ be a strict formula over~$p$.  
Denote by $\Phi\hm$ the formula over $\Hm(p)\in\mcH_{\bfd^*}[q+1]$ given by
\[
  \Phi\hm:=\Phi(p_1\hm,\ldots,p_q\hm)\wedge \big(\|p\|X_0 > 0 \big).
\]
Then the sets $W(p,\Phi)$ and $\Ap(\Hm(p),\Phi\hm))$ are homeomorphic.
\eproof
\end{prop}

\subsection{Estimation of the condition number}
\label{sec:kappa}

In~\cite[\S3.4]{BCTC1} we defined a condition number
$\kappabar(f)\in [1,\infty]$ associated to a tuple $f\in\Hd[q]$, 
whose inverse measures how near are the intersections between the 
hypersurfaces given by $f$ from being non-transversal.  The condition
number $\kappaff(p)$ in Theorem~\ref{thm:main_result} 
was then defined~\cite[\S7.1]{BCTC1} to be $\kappabar(\Hm(p))$).  
The quantity $\kappabar(f)$ provides information on the geometry of every possible 
spherical semialgebraic set built from $f$. Tuples $f$ for which 
$\kappabar(f)=\infty$ are said to be {\em ill-posed}. They are precisely those 
tuples for which there exists a formula $\Phi$ such that arbitrary 
small perturbations of $f$ may change the topology of $\Ap(f,\Phi)$. 
The set $\overline{\Sigma}_{\bfd}[q]$ of ill-posed tuples has positive codimension in 
$\Hd[q]$ and $\kappabar(f)$ estimates how far is $f$ from $\overline{\Sigma}_{\bfd}[q]$. 

The first substantial computational effort performed by {\sf Homology} 
is to estimate the condition number $\kappabar(f)$ of a tuple 
$f\in\Hd[q]$. The following result, Proposition~2.2 in~\cite{BCTC1},  
deals with this task.

\begin{prop}\label{prop:kappa-est}
There is an algorithm {\sc $\kappabar$-Estimate}
that, given $f\in\Hd[q]$, returns a number~$\sfK$ such that 
$$
   0.99\,\kappabar(f) \ \leq\ \sfK \ \leq\ \kappabar(f) 
$$
if $\kappabar(f)<\infty$,  
or loops forever otherwise. The cost of this algorithm is bounded 
by $\big(qnD\kappabar(f)\big)^{\Oh(n)}$. \eproof
\end{prop} 
When $\kappabar(f)$ is infinity the algorithm loops forever because it cannot find an upper bound for $\kappabar(f)$. However, we note that the algorithm in~\cite{BCTC1} can be modified to stop if $\kappabar(f)$ is too large and return this fact (see~\cite[Proposition~6.3]{BCTC1}).

\subsection{The Gabrielov-Vorobjov construction}

The main  idea behind the algorithm in~\cite{BCTC1} consists of finding 
a finite collection of points $\mcX$ and a radius $\varepsilon$ 
such that the union $\cup_{x\in\mcX}B(x,\varepsilon)$ contains the 
set $\Ap(f,\Phi)$ and continuously retracts to it. 
Out of the realm of lax formulas we dealt with in~\cite{BCTC1} this idea 
becomes impracticable. The reason is 
that the connected components of $\Ap(f,\Phi)$ 
may now not be separated by a positive distance. 
Consider for instance the pair $f=(X-Y,Y)$ and the semialgebraic set 
given by 
\begin{equation}\label{eq:example}
    \Phi\equiv (X-Y=0 \wedge Y>0) \vee (Y=0 \wedge X-Y>0). 
\end{equation}
The set $\Ap(f,\Phi)$ consists of two open half-lines with origin at $(0,0)$.  
Balls close to $(0,0)$ containing initial segments of the two half-lines are 
likely to intersect. 

To circumvent this problem we will rely on a beautiful construction 
conceived by A.~Gabrielov and N.~Vorobjov in~\cite{gabrivorob} that 
produces closed semialgebraic approximations to semialgebraic sets. 
These are obtained by combining relaxations of the equalities and 
strengthenings of the inequalities in the formula $\Phi$. 
These relaxations and strengthenings can be seen as a rewriting of the 
formula $\Phi(f)$ into a new formula. 

\begin{defi}
Given a monotone formula $\Phi$ over $f\in \Hd[q]$ and positive $\delta$ and $\varepsilon$, 
the \textit{Gabrielov-Vorobjov $(\delta,\varepsilon)$-block 
$\gv_{\delta,\varepsilon}(f,\Phi)$} is the spherical semialgebraic set defined by 
the following rewriting of $\Phi(f)$,
\begin{align*}
  &f_i=0 \rew |f_i(x)|\leq \varepsilon\|f_i\|,\\ 
  &f_i>0 \rew f_i(x)\geq \delta\|f_i\|, \mbox{ and}\\ 
  &f_i<0 \rew f_i(x)\leq -\delta\|f_i\|.
\end{align*}
 
Given $\bfdelta,\bfeps\in (0,\infty)^m$, the 
\textit{Gabrielov-Vorobjov $(\bfdelta,\bfeps)$-approximation 
$\gv_{\bfdelta,\bfeps}(f,\Phi)$ (of order $m$) 
of $\Ap(f,\Phi)$} is 
the spherical semialgebraic set given by
\begin{equation}\label{eq:GVapprox}
\gv_{\bfdelta,\bfeps}(f,\Phi):=\bigcup_{k=1}^m\gv_{\delta_k,\varepsilon_k}(f,\Phi)\text{.}
\end{equation}
\end{defi}

Note that Gabrielov-Vorobjov blocks and Gabrielov-Vorobjov
approximations are compact subsets of $\bbS^n$. The norms $\|f_i\|$ in the definition above 
are not in~\cite{gabrivorob}. We have added them here as they 
make clearer statements in our context.

The main result of~\cite{gabrivorob}, Theorem~1.10 there, yields the following immediate consequence (which also holds with our modified definition of Gabrielov-Vorobjov blocks). Recall that $\pi_k$ stands for the $k$th homotopy group and $H_k$ for the $k$th homology group.

\begin{theo}[Gabrielov-Vorobjov Theorem]\label{thm:GV}
Let $f\in\Hd[q]$, $\Phi$ be a monotone formula over $f$, 
$m\in\bbN$, and 
$\bfdelta,\bfeps\in(0,\infty)^m$. If
\begin{equation}\label{GVcondition}
0<\varepsilon_1\ll\delta_1\ll\cdots\ll\varepsilon_m\ll\delta_m\ll 1\text{,}
\end{equation}
then, for $k\in\{0,\ldots,m-1\}$, there are homomorphisms
\[\phi_k:\pi_k(\gv_{\bfdelta,\bfeps}(f,\Phi))\rightarrow \pi_k(\Ap(f,\Phi))\]
and
\[\varphi_k:H_k(\gv_{\bfdelta,\bfeps}(f,\Phi))\rightarrow H_k(\Ap(f,\Phi))\]
that are isomorphisms for $k< m-1$ and epimorphisms when $k=m-1$.
\eproof
\end{theo}

In this statement, the relations 
$0<a_1\ll\cdots\ll a_t\ll 1$ of reals $a_i$ mean that 
there are functions 
$h_k:(0,1)^{t-k}\rightarrow (0,1)$ 
such that $0<a_k<h_k(a_{k+1},\ldots,a_t)$ for all~$k$. 

\begin{remark}
Homotopy groups (without specifying a base point) are only defined 
for connected spaces. However, the bijection between 
$\pi_0(\gv_{\bfdelta,\bfeps}(f,\Phi))$ and $\pi_0(\Ap(f,\Phi))$ identifies 
the connected components of $\gv_{\bfdelta,\bfeps}(f,\Phi)$ and 
$\Ap(f,\Phi)$. Therefore we can naturally interpret 
$\phi_k:\pi_k(\gv_{\bfdelta,\bfeps}(f,\Phi))\rightarrow \pi_k(\Ap(f,\Phi))$, for $k>0$, 
as the family of maps
\[
\{\phi_k:\pi_k(C)\rightarrow \pi_k(\phi_0(C))\mid C\in\pi_0(\gv_{\bfdelta,\bfeps}(f,\Phi))\}.
\]
The assumption of connectedness 
in~\cite{gabrivorob} is only for technical ease of the exposition.
\end{remark}

The proof of Theorem~\ref{thm:GV}
is an elegant conjunction of geometric insight and technical skill. 
While it is out of our reach to explain the ideas behind it (the interested 
reader will find these ideas in~\cite{gabrivorob}) we believe a few simple 
examples may provide some intuition. 

\begin{exam}
Consider the pair $f$ and formula $\Phi$ in~\eqref{eq:example} we started this subsection with. 
For any pair $(\delta,\varepsilon)$ with $0<\varepsilon<\delta$ the block 
$\gv_{\delta,\varepsilon}(f,\Phi)$ is given by
$$
   \Big(|X-Y|\le \varepsilon\sqrt{2} \wedge (Y\ge\delta) \Big)
   \vee \Big( |Y|\le\varepsilon \wedge 
   \big(X-Y\ge-\delta\sqrt{2}\big)\Big) 
$$
and looks as in Figure~\ref{fig:2lines}.

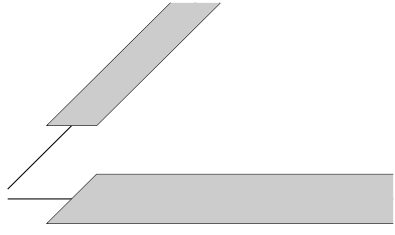
\begin{figure}[ht]
\begin{center}
\begin{tikzpicture}[scale=1.3,point/.style={draw,minimum size=0pt,
    inner sep=1pt,circle,fill=black}]
\draw (0.1,0) -- (4,0);
\draw (4,-0.25) -- (0.5,-0.25) -- (1,0.25) -- (4,0.25);
\fill[black!20!white] (4,0.25) -- (4,-0.25) -- (0.5,-0.25) -- (1,0.25);
\draw (0.1,0.1) -- (2,2);
\draw (1.75,2) --(0.5,0.75)--(1.0,0.75) -- (2.25,2);
\fill[black!20!white](0.5,0.75) -- (1.0,0.75) -- (2.25,2) -- (1.75,2);
\end{tikzpicture}
\end{center}
\caption{The Gabrielov-Vorobjov construction for two open half-lines}\label{fig:2lines}
\end{figure}
It is clear that this block is homotopically equivalent to $\Ap(f,\Phi)$.
\end{exam}

\begin{exam}
The number~$m$ of blocks needed in the Gabrielov-Vorobjov construction to recover 
the $k$th homology group of $\Ap(f,\Phi)$ may reach the bound $k+2$ in Theorem~\ref{thm:GV}. 
Let $f=(X,Y)$ and consider 
$$
  \Phi\equiv (X=0 \wedge Y=0) \vee (X=0 \wedge Y>0) \vee (X>0 \wedge Y=0) \vee 
  (X>0 \wedge Y>0)
$$
so that $\Ap(f,\Phi)$ is the closed positive quadrant. Now take any sequence 
$$
   0<\varepsilon_1<\delta_1<\varepsilon_2<\delta_2<\varepsilon_3<\delta_3.
$$
At the left of Figure~\ref{fig:dim-opt} we see in light grey shading the 
block $\gv_{\delta_1,\varepsilon_1}(f,\Phi)$. It is not connected; not even 
the 0th homology group is correct. At the center of the figure we see that 
same first block with $\gv_{\delta_2,\varepsilon_2}(f,\Phi)$ superimposed in a
darker shade of grey. Now the union of the first two blocks is connected (so 
$H_0$ is correct) but not simply connected: the first homology group is wrong.
We obtain a contractible set, homotopically equivalent to $\Ap(f,\Phi)$, when 
we add the third block, at the right of the figure, to the union. 

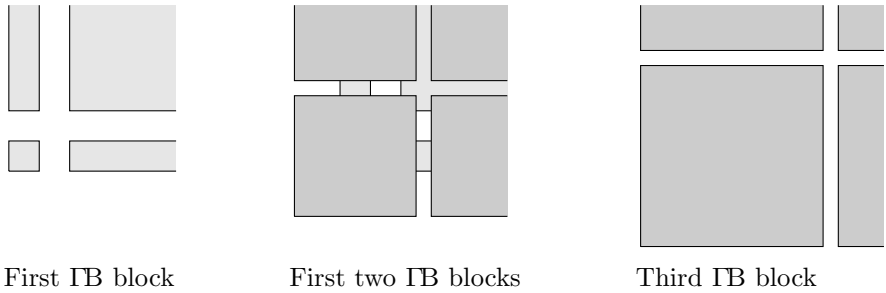
\begin{figure}[h]
\begin{center}
\begin{tikzpicture}[scale=0.2,point/.style={draw,minimum size=0pt,
    inner sep=1pt,circle,fill=black}]
\fill[black!10!white] (-1,-1) -- (-1,1) -- (1,1) -- (1,-1);
\draw (-1,-1) -- (-1,1) -- (1,1) -- (1,-1) -- (-1,-1);
\fill[black!10!white] (10,-1) -- (3,-1) -- (3,1) -- (10,1);
\draw (10,-1) -- (3,-1) -- (3,1) -- (10,1);
\fill[black!10!white] (-1,10) -- (-1,3) -- (1,3) -- (1,10);
\draw (-1,10) -- (-1,3) -- (1,3) -- (1,10);
\fill[black!10!white] (3,10) -- (3,3) -- (10,3) -- (10,10);
\draw (3,10) -- (3,3) -- (10,3);
\node[anchor=west] (note1) at (-2,-8) {\small First $\gv$ block};
\end{tikzpicture}
\hspace{1cm}
\begin{tikzpicture}[scale=0.2,point/.style={draw,minimum size=0pt,
    inner sep=1pt,circle,fill=black}]
    \fill[black!10!white] (-1,-1) -- (-1,1) -- (1,1) -- (1,-1);
\draw (-1,-1) -- (-1,1) -- (1,1) -- (1,-1) -- (-1,-1);
\fill[black!10!white] (10,-1) -- (3,-1) -- (3,1) -- (10,1);
\draw (10,-1) -- (3,-1) -- (3,1) -- (10,1);
\fill[black!10!white] (-1,10) -- (-1,3) -- (1,3) -- (1,10);
\draw (-1,10) -- (-1,3) -- (1,3) -- (1,10);
\fill[black!10!white] (3,10) -- (3,3) -- (10,3) -- (10,10);
\draw (3,10) -- (3,3) -- (10,3);
\fill[black!20!white] (-4,-4) -- (-4,4) -- (4,4) -- (4,-4);
\draw (-4,-4) -- (-4,4) -- (4,4) -- (4,-4) -- (-4,-4);
\fill[black!20!white] (10,-4) -- (5,-4) -- (5,4) -- (10,4);
\draw (10,-4) -- (5,-4) -- (5,4) -- (10,4);
\fill[black!20!white] (-4,10) -- (-4,5) -- (4,5) -- (4,10);
\draw (-4,10) -- (-4,5) -- (4,5) -- (4,10);
\fill[black!20!white] (5,10) -- (5,5) -- (10,5) -- (10,10);
\draw (5,10) -- (5,5) -- (10,5);
\node[anchor=west] (note1) at (-5,-8) {\small First two $\gv$ blocks};
\end{tikzpicture}
\hspace{1cm}
\begin{tikzpicture}[scale=0.2,point/.style={draw,minimum size=0pt,
    inner sep=1pt,circle,fill=black}]
\fill[black!20!white] (-6,-6) -- (-6,6) -- (6,6) -- (6,-6);
\draw (-6,-6) -- (-6,6) -- (6,6) -- (6,-6) -- (-6,-6);
\fill[black!20!white] (10,-6) -- (7,-6) -- (7,6) -- (10,6);
\draw (10,-6) -- (7,-6) -- (7,6) -- (10,6);
\fill[black!20!white] (-6,10) -- (-6,7) -- (6,7) -- (6,10);
\draw (-6,10) -- (-6,7) -- (6,7) -- (6,10);
\fill[black!20!white] (7,10) -- (7,7) -- (10,7) -- (10,10);
\draw (7,10) -- (7,7) -- (10,7);
\node[anchor=west] (note1) at (-7,-8) {\small Third $\gv$ block};
\end{tikzpicture}
\end{center}
\caption{The Gabrielov-Vorobjov construction for the positive quadrant}\label{fig:dim-opt}
\end{figure}
\end{exam}

We now remark that no explicit form of the functions $h_k$ behind the 
relations $0<a_1\ll\cdots\ll a_t\ll 1$ is given in~\cite{gabrivorob}. 
Our first main result, Theorem~\ref{thm:quantitativeGV} below, provides a very 
simple answer to this issue for well-posed tuples of polynomials.

\begin{theo}[Quantitative Gabrielov-Vorobjov Theorem]\label{thm:quantitativeGV}
In Theorem~\ref{thm:GV}, condition~\eqref{GVcondition} can be replaced by
\begin{equation}\label{GVconditionexplicit}
0<\varepsilon_1<\delta_1<\cdots<\varepsilon_m<\delta_m <\frac{1}{\sqrt{2}\kappabar(f)}
\end{equation}
when $\kappabar(f)<\infty$.
\end{theo}

\begin{exam}
The simple form of the inequalities in~\eqref{GVconditionexplicit} requires 
well-posedness, i.e., $\kappabar(f)<\infty$. To see this, consider $f=(X,Y,X-Y)$ and 
\[
\Phi\equiv ((X=0)\wedge(Y>0))\vee((X-Y=0)\wedge(Y>0)).
\]
The set $\Ap(f,\Phi)$ consists of two half-lines with a common origin but 
without this origin. Note that $\kappabar(f)=\infty$. Figure~\ref{fig:ill-posed} 
shows $\Ap(f,\Phi)$ at the left. 
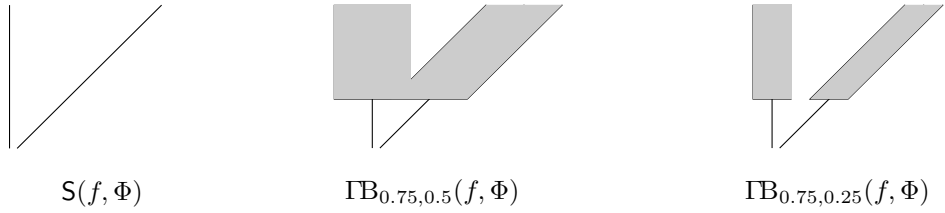
\begin{figure}[h]
\begin{center}
\begin{tikzpicture}[scale=1,point/.style={draw,minimum size=0pt,
    inner sep=1pt,circle,fill=black}]
\draw (0,0.1) -- (0,2);
\draw (0.1,0.1) -- (2,2);
\node[anchor=west] (note1) at (0.55,-0.5) {\small $\Ap(f,\Phi)$};
\end{tikzpicture}
\hspace{2cm}
\begin{tikzpicture}[scale=1,point/.style={draw,minimum size=0pt,
    inner sep=1pt,circle,fill=black}]
\draw (0,0.1) -- (0,2);
\draw (-0.5,2) -- (-0.5,0.75) -- (1.25,0.75) -- (2.5,2);
\draw (0.5,2)  -- (0.5,1)-- (1.5,2);
\fill[black!20!white] (0.5,2) -- (-0.5,2) -- (-0.5,0.75) -- (0.5,0.75);
\draw (0.1,0.1) -- (2,2);
\fill[black!20!white] (0.25,0.75) -- (1.25,0.75) -- (2.5,2) -- (1.5,2);
\node[anchor=west] (note1) at (-0.5,-0.5) {\small $\gv_{0.75,0.5}(f,\Phi)$};
\end{tikzpicture}
\hspace{2cm}
\begin{tikzpicture}[scale=1,point/.style={draw,minimum size=0pt,
    inner sep=1pt,circle,fill=black}]
\draw (0,0.1) -- (0,2);
\draw (-0.25,2) -- (-0.25,0.75) -- (0.25,0.75) -- (0.25,2);
\fill[black!20!white]  (0.25,0.75) -- (0.25,2) -- (-0.25,2) -- (-0.25,0.75);
\draw (0.1,0.1) -- (2,2);
\draw  (1.75,2) -- (0.5,0.75) -- (1.0,0.75) -- (2.25,2);
\fill[black!20!white](0.5,0.75) -- (1.0,0.75) -- (2.25,2) -- (1.75,2);
\node[anchor=west] (note1) at (-0.5,-0.5) {\small $\gv_{0.75,0.25}(f,\Phi)$};
\end{tikzpicture}
\end{center}
\caption{The Gabrielov-Vorobjov construction for an ill-posed system}\label{fig:ill-posed}
\end{figure}
The center and right parts of the figure 
exhibit two Gabrielov-Vorobjov Approximations for it with $m=1$ but different 
pairs $(\delta,\varepsilon)$. 
The middle part shows that the condition $\varepsilon<\delta$ is not strong enough 
to guarantee the conclusions of Theorem~\ref{thm:GV} for $m=1$. An easy computation 
shows that, in this case, we need $0<\varepsilon<\delta/2$ (as in the right part 
of the figure).
\end{exam}

\subsection{Reach and condition of Gabrielov-Vorobjov approximations}
\label{se:reach-cond}

In the following let $f\in\Hd[q]$ be such that $\kappabar(f) < \infty$. 
Moreover let $\Phi$ be a strict formula over $f$ and the integer $m$ satisfy 
$m \ge 2 +\dim\Ap(f,\Phi)$; such $m$ can be easily computed since $f$
is well-posed.
Moreover, let $\sfK$ be an estimate of $\kappabar(f)$ 
as in Proposition~\ref{prop:kappa-est}.
Using the increasing tuple of positive reals 
\begin{equation}\label{eq:epds}
  \big(\varepsilon_1,\delta_1,\ldots,\varepsilon_m,\delta_m\big)
   := \big( 15(2m+1)D^2\sfK^2 \big)^{-1}\,\big(1,2,\ldots, 2m \big) ,
\end{equation}
we define the Gabrielov-Vorobjov approximation of $\Ap(f,\Phi)$ 
\begin{equation}\label{eq:def-GB}
 \gv(f,\Phi) : = \gv_{\bfdelta,\bfeps}(f,\Phi) .
\end{equation}  
By Theorem~\ref{thm:quantitativeGV}, 
$\gv(f,\Phi)$ and $\Ap(f,\Phi)$ have the same homology. 

As $\gv(f,\Phi)$ is closed, a first idea would be to apply to it 
the algorithm developed in~\cite{BCTC1}.
Unfortunately, as we are about to see, this idea does not work. 
Let $e:=4m$ and abbreviate 
\begin{equation}\label{eq:seq_t}
   t:=(t_1,\ldots,t_e):=
   (\varepsilon_1,\delta_1,\ldots,\varepsilon_m,\delta_m,
         -\varepsilon_1,-\delta_1,\ldots,-\varepsilon_m,-\delta_m) ,
\end{equation}
where the $\varepsilon_j,\delta_j$ are those 
in~\eqref{eq:epds}. 
We note that 
$\gv(f,\Phi)$ is defined by a Boolean formula~$\Phibar$ 
(depending on $\Phi$) 
in terms of the tuple of polynomials
\begin{equation}\label{eq:def-fbar}
\fbar := \{f_i-t_j\|f_i\|\}_{i\le q,j\le e} .
\end{equation}
But the polynomials in $\fbar$ are no longer homogeneous. 
It is easy to verify that, in general, $\kappaff(\fbar)=\infty$, since the $f_i-t_j\|f_i\|$, for different $j$, intersect tangentially at infinity whenever $f_i$ is non-linear. Thus it is hopeless 
to pass $\fbar$ as input to the algorithm described and analyzed in~\cite{BCTC1}. 

A closer look to this analysis reveals however, that its main ideas can be 
reproduced in our situation. The first stepping stone in~\cite{BCTC1} towards 
the algorithm's design is the following result (Theorem~2.3 there). 

\begin{theo}[Basic Homotopy Witness Theorem I]\label{thm:NSW}
Let $f\in\Hd[q]$ and $\phi$ be a purely conjunctive lax formula over $f$.
Moreover, let $\mcX\subseteq\bbS^n$ be a closed subset 
and $\varepsilon>0$ be such that
$$
   3 d_H\big(\mcX, \Ap(f,\phi)\big) 
   < \varepsilon <\frac{1}{14 D^{\frac32}\kappabar(f)}.
$$ 
Then the inclusion $\Ap(f,\phi) \hookrightarrow \mcU(\mcX,\varepsilon)$ 
induces a homotopy equivalence.\eproof
\end{theo}

In this statement, $d_H$ denotes the {\em Hausdorff distance} between 
two nonempty compact sets $W,V\subseteq\bbR^{n+1}$ which, we recall, 
is given by
$$
   d_H(W,V):=\max\Big\{\max_{v\in V} d(W,v),
   \max_{w\in W} d(w,V)\Big\}
$$
where $d$ denotes Euclidean distance in $\bbR^{n+1}$. 
If either $V$ or $W$ is empty then one takes $d_H(V,W):=\infty$, unless both are empty,
in which case $d_H(V,W):=0$. Also, 
\begin{equation}\label{eq:U}
\mcU(\mcX,r):=\bigcup_{x\in\mcX} B(x,r)
\end{equation} 
denotes the {\em open $r$-neighborhood of $\mcX$} in $\bbR^{n+1}$ 
and $B(x,r)$ the Euclidean open ball 
with center $x$ and radius~$r$. 

Theorem~\ref{thm:NSW}  goes back to a 
result by Niyogi, Smale and Weinberger~\cite{NiSmWe2008} proving the statement 
for manifolds and expressing the bound on the right-hand side in terms of the 
reach of the manifold. For submanifolds $\mcM$ of $\bbS^n$ given as the zero set of a 
tuple $f$ of homogeneous polynomials, this reach was bounded 
in~\cite{CKS16} in  terms 
of $\max_{x\in\mcM}\gamma(f,x)$ where $\gamma(f,x)$ is the $\gamma$-invariant 
introduced by Smale~\cite{Smale86}. This result was given a much simpler 
proof in~\cite[Theorem~2.11]{BCL17}. The Higher Derivative 
Estimate~\cite{Bez4,Condition} further allowed one to replace $\gamma$ by 
the condition number. 

In Section~\ref{sec:gamma} we will take advantage of this sequence of results. 
Basically, we will show that the $\gamma$ at points of algebraic sets defined using 
polynomials from $\fbar$ can be bounded in terms of the $\gamma$ for $f$ instead 
and, hence, ultimately in terms of $\kappabar(f)$. We will then derive 
an extension of Theorem~\ref{thm:NSW} appropriate to our setting. 

It is convenient to introduce some more terminology. 

\begin{defi}\label{def:laxF} 
Let $f\in\Hd[q]$ and $t\in\bbR^e$. 
We call a \emph{lax formula} $\Phi$ over $(f,t)$ 
a monotone formula,  
whose atoms are of the form $(f_i\geq t_j\|f_i\|)$ or $(f_i\leq t_j\|f_i\|)$, 
with $i\le q$ and $j\le e$ and 
denote by $\Ap(f,t,\Phi)$ the semialgebraic subset of the sphere~$\bbS^n$ 
described by these lax inequalities. 
\end{defi}

\begin{remark}
The set $\Ap(f,t,\Phi)$ is defined in terms of sign conditions on the polynomials 
in $\fbar$ (see~\eqref{eq:def-fbar}), which are not homogeneous.
Every intersection of Gabrielov-Vorobjov blocks 
$\gv_{\delta_j,\varepsilon_j}(f,\phi)$ is of the form $\Ap(f,t,\phi)$ for a purely conjunctive formula $\phi$. 
Moreover, the formula $\Phibar$ defined right before~\eqref{eq:def-fbar}
can be viewed as a formula over $(f,t)$ such that
$\gv(f,\Phi)=\Ap(f,t,\Phibar)$, by the construction of the
Gabrielov-Vorobjov approximations.
\end{remark}

\begin{theo}[Generalized Basic Homotopy 
Theorem]\label{theo:homotopywitness}
Let $f\in\Hd[q]$, $T>0$ be such that $\sqrt{2}\,\kappabar(f)T<1$, $t\in(-T,T)^e$, and 
$\phi$ a purely conjunctive lax formula over $(f,t)$.
Moreover, let $\mcX\subseteq\bbS^n$ be a closed subset 
and $\varepsilon>0$ be such that
$$
   3 d_H\big(\mcX, \Ap(f,t,\phi)\big) 
   < \varepsilon <\frac{1}{48 D^{\frac32}\kappabar(f)}.
$$ 
Then the inclusion $\Ap(f,t,\phi) \hookrightarrow \mcU(\mcX,\varepsilon)$ 
is a homotopy equivalence.
\end{theo}

We prove this result in~\S\ref{sec:NSW}. 

\subsection{Point clouds and cell complexes}

For a finite set $\mcX\subseteq \bbR^{n+1}$, we compute the homology groups 
of 
$\mcU(\mcX,\varepsilon)$ by computing those of the 
{\em \v{C}ech complex} 
$\cech{\varepsilon}{\mcX}$ which is the simplicial complex whose $k$-faces
are the sets of $k+1$ points $\{x_0,\ldots,x_k\}\subseteq\mcX$ such that
$\cap_{i\leq k} B(x_i,\varepsilon)\neq \varnothing$. The justification of this is the 
Nerve Theorem~\cite[Corollary~4G.3]{hatcher} that guarantees that these two 
spaces are homotopically equivalent.

Theorem~\ref{theo:homotopywitness} opens the way to find a good pair $(\mcX,\varepsilon)$ 
for basic sets. For non-basic sets, one needs a variation, 
where the simplicial complex is constructed in a more complicated way. We recall 
this construction. 
For given $f\in\Hd[q]$ and $t\in\bbR^e$ consider
the $2qe$ {\em atomic} sets 
\begin{align}\label{eq:atomic}
   S_{i,j}^\leq &:=\Ap(f_i-t_j\|f_i\|\leq 0),\\ 
   S_{i,j}^\geq &:=\Ap(f_i-t_j\|f_i\|\geq 0),\nonumber
\end{align}
for $i\le q$ and $j\le e$. Furthermore, 
also for $i\le q$ and $j\le e$, let  
$\mcX_{i,j}^{\leq}$ and $\mcX_{i,j}^{\geq}$ be finite subsets of $\bbS^n$.
Given a lax formula $\Psi$ over $(f,t)$, we 
can abuse notation and write $\Phi(\mcX_{i,j}^{\propto})$
to denote the finite set recursively built from the 
sets $\mcX_{i,j}^{\propto}$ in the same way $\Psi$  
is built from the expressions $(f_i\propto t_j\|f_i\|)$, taking unions 
when we find $\vee$ and intersections when we find $\wedge$.
As unions and intersections of the simplicial complexes
$\cech{\varepsilon}{\mcX_{i,j}^{\propto_j}}$ are well defined, 
we can further abuse notation to denote by 
\[
   \Psi\left(\cech{\varepsilon}{\mcX_{i,j}^\propto}\mid i\in[q],\,j\in[e],\, \propto\in\{\leq,\geq\}\!\right)\subseteq \Delta^{\cup_{i,j,\propto}\mcX_{i,j}^{\propto}},
\]
the simplicial complex recursively built from  
the $\cech{\varepsilon}{\mcX_{i,j}^{\propto_j}}$. This is 
explained with more details in~\cite[\S2.4]{BCTC1}. 
Applying $\Psi$ to a collection of point clouds or to one 
of simplicial complexes is an abuse of notation but it allows us 
to clearly state the following extension of Theorem~\ref{theo:homotopywitness}. 

\begin{theo}[Generalized Homology Witness Theorem]\label{theo:homologywitness} 
Let $f\in\Hd[q]$, $\sfK$ be as in Proposition~\ref{prop:kappa-est}, and $t\in\bbR^{4m}$ as in~\eqref{eq:seq_t}. 
Let $\varepsilon>0$,
and assume that for $i\in[q]$ and $j\in[4m]$, 
$\mcX_{i,j}^\leq$ and $\mcX_{i,j}^\geq$ are finite subsets of $\bbS^n$
such that for all purely conjunctive lax formulas $\phi$ over $(f,t)$, we have 
\begin{equation}\label{eq:hypoGHW}
   3 d_H\left(\phi\left(\mcX_{i,j}^{\propto}\mid i\in[q],\,j\in[4m],\,\propto\in\{\leq,\geq\}\right), 
   \Ap(f,t,\phi)\right) 
   \ < \ \varepsilon \ <\ \frac{1}{180(2m+1)D^{5/2}\sfK^2}. 
\end{equation}
Then, for all strict formulas $\Phi$ over $f$, 
the Gabrielov-Vorobjov approximation $\gv(f,\Phi)$ defined in~\eqref{eq:def-GB}
and the simplicial complex 
$$
 \fkC =\Phibar\Big(\cech{\varepsilon}{\mcX_{i,j}^\propto}\mid i\in[q],\,j\in[4m],\, \propto\in\{\leq,\geq\}\!\Big)
$$ 
have the same homology.
\end{theo}

We prove this result in~\S\ref{sec:HWT}.

To make use of Theorem~\ref{theo:homologywitness} we need 
to construct finite sets 
$\mcX_{i,j}^\leq$ and $\mcX_{i,j}^\geq$ satisfying~\eqref{eq:hypoGHW}.
As in~\cite{BCTC1}, to do this, we will rely on the notion of an algebraic neighborhood 
(see \S\ref{se:sample-rel}), 
now adapted to our context, and on the use of grids. 

For $\ell\in\bbN$ one can construct a grid $\mcG_\ell\subseteq\bbS^n$ such that 
$$
   |\mcG_\ell|\leq (n2^\ell)^{\Oh(n)}
   \qquad\mbox{ and }\qquad
   \bbS^n\subseteq\cup_{x\in\mcG_\ell}B_{\bbS}(x,r_\ell)\subset
   \cup_{x\in\mcG_\ell}B(x,r_\ell) ,
$$
where $r_\ell:=2^{-\ell}$ (see~\cite[\S6.1]{BCTC1}) and 
$B_{\bbS}(x,r_\ell)$ is the ball centered at $x$ of radius $r_\ell$ 
with respect to the Riemannian distance on the sphere. 

Given $f\in\Hd[q]$, $t\in\bbR^{e}$, 
for each $i\in[q]$, $j\in[e]$, 
we consider the set $\mcX_{i,j}^{\ge}$ of grid points~$x$ in $\mcG_\ell$ satisfying 
$$
  f_i(x) > (t_j + D^{\frac12}r_\ell) \|f_i\| ,
$$
which means that the inequality $f_i > t_j \|f_i\|$ holds true with the tolerance $ D^{\frac12}r_\ell$. 
Similarly, we define the set of of grid points $\mcX_{i,j}^{\le}$. 

In~\S\ref{sec:samp} we prove the following result (generalizing Theorem~6.5 in~\cite{BCL17}). 

\begin{theo}[Generalized Sampling Theorem]\label{thm:samp}
Let $f\in\Hd[q]$, $\sfK$ be as in Proposition~\ref{prop:kappa-est},
and $t\in\bbR^{4m}$ as in~\eqref{eq:seq_t}. 
Moreover, let $\ell\in\bbN$ and $\rho\in(0,1]$ be chosen such that 
$r_\ell\leq \frac{\rho}{1620(2m+1)D^{3}\sfK^3}$.
Then, for all purely conjunctive lax formulas $\phi$ over $(f,t)$, we have 
$$
   3 d_H\left(\phi\left(\mcX_{i,j}^{\propto}\mid i\in[q],\,j\in[4m],\,\propto\in\{\leq,\geq\}\right), 
   \Ap(f,t,\phi)\right) 
   \ \leq\ \frac{\rho}{180(2m+1)D^{5/2}\sfK^2}.
$$
\end{theo}

\subsection{The algorithm}

We can finally combine all the previous ideas in an algorithm and prove our 
main result. 

\begin{minipage}[t]{0.9\textwidth}
\begin{algorithm*}[H]
\DontPrintSemicolon
\SetKwInOut{input}{Input}
\SetKwInOut{output}{Output}
\caption{\textsc{Homology}}\label{alg:homology}
\input{$p\in\Pd[q]$\\
Boolean formula $\Psi$ over $p$\\
}
\hrulefill

eliminate negations and lax inequalities in $\Psi$\;
$f \leftarrow \Hm(p)$\;
$\Phi\leftarrow\Psi\hm$\;
$\sfK\leftarrow \mbox{{\sc $\kappabar$-Estimate}}(f)$\;
compute the sequence $\varepsilon_1,\delta_1,\ldots,\varepsilon_m,\delta_m$ in~\eqref{eq:epds}\;
$\ell\leftarrow \lceil\log_2 1638 (2m+1)D^{3} \sfK^3\rceil$\;
$\varepsilon\leftarrow 1/(181\,(2m+1)D^{5/2} \sfK^2)$\;
\For{$i = 1,\ldots,q$ {\bf and} $j=1,\ldots,e$}{
compute $\mcX_{i,j}^{\leq}$ and 
$\fkA_{i,j}^{\leq}\leftarrow \{\sigma\in\cech{\varepsilon}{\mcX_{i,j}^{\leq}}\mid |\sigma|\leq n+2\}$\;
compute $\mcX_{i,j}^{\geq}$ and 
$\fkA_{i,j}^{\geq}\leftarrow \{\sigma\in\cech{\varepsilon}{\mcX_{i,j}^{\geq}}\mid |\sigma|\leq n+2\}$\;
}
$\fkC\leftarrow \Phibar\left(\fkA_{1,1}^{\leq},\fkA_{1,1}^{\geq},\ldots,\fkA_{q,e}^{\leq},\fkA_{q,e}^{\geq}\right)$\;
compute the homology groups $H_*$ of $\fkC$ of order up to $n$\;
\KwRet{$H_*$}\\
\hrulefill\\
\output{Sequence of groups $H_*$}
\postcondition{$H_*$ are the homology groups of $\Ap(p,\Phi)$.}
\end{algorithm*}
\end{minipage}
\medskip

\begin{proof}[Proof of Theorem~\ref{thm:main_result}]
The computation of $H_*$ from $\fkC$ is described with 
details in~\cite[Proposition~4.3]{CKS16}. 

As for the correctness of the algorithm. 
The choice of $\ell$ guarantees that 
\[
r_{\ell}=2^{-\ell}\leq \frac{1}{1638(2m+1)D^3\sfK^3}=\frac{1620/1638}{1620(2m+1)D^3\sfK^3}.
\]
Thus Theorem~\ref{thm:samp} applies 
with $\rho=1620/1638$ and gives the upper bound $\frac{1}{182(2m+1)D^{5/2}\sfK^2}$. 
By Theorem~\ref{theo:homologywitness}, it is enough then to choose $\varepsilon$ inside
\[\left(\frac{1}{182(2m+1)D^{5/2}\sfK^2},\frac{1}{180(2m+1)D^{5/2}\sfK^2}\right).\]
Our choice of~$\varepsilon$ satisfies these bounds.
Hence, we have the homology isomorphisms 
$$
  H_*(\fkC)\simeq H_*(\gv(f,\Phi))\simeq H_*(\Ap(f,\Phi)) 
  \simeq H_*(W(p,\Psi)),
$$
the first by Theorem~\ref{theo:homologywitness}, 
the second by Theorem~\ref{thm:quantitativeGV} 
and our choice of $\varepsilon_1,\delta_1,\ldots,\varepsilon_m,\delta_m$, 
and the last by 
Proposition~\ref{generaltospherical}. We note that we only care about the homology groups up to order $n$, since higher homology groups of a semialgebraic set in $\bbR^n$ vanish.

Finally, the cost analysis is the same as that 
in~\cite[\S7.3]{BCTC1}.
\end{proof}

\section{Quantitative Gabrielov-Vorobjov Theorem}\label{sec:prrofGV}

The objective of this section is to prove the Quantitative Gabrielov-Vorobjov Theorem~\ref{thm:quantitativeGV}.

The idea of the proof is simple, we transform a pair $(\bfdelta,\bfeps)$ 
satisfying~\eqref{GVconditionexplicit} into a pair $(\bfdelta',\bfeps')$
satisfying~\eqref{GVcondition} so that Theorem~\ref{thm:GV} can be applied 
and we show that in doing so, the homotopy type of the associated 
Gabrielov-Vorobjov set remains unchanged. To simplify the argument, we 
proceed by steps, modifying only one component in the pair $(\bfdelta,\bfeps)$ 
at a time. In the next subsection we describe these basic reductions, leading 
us to a statement, Proposition~\ref{thm:reductionGV2}, which implies 
Theorem~\ref{thm:quantitativeGV}. Then, 
in~\S\ref{sec:difTop} we recall some fundamental notions from differential 
topology which we use in~\S\ref{sec:proof_red} to prove 
Proposition~\ref{thm:reductionGV2}.

\subsection{Basic reductions}

We write $(\bfdelta,\bfeps)\leq_{\syd, i}(\bfdelta',\bfeps')$ when
$\bfeps=\bfeps'$, $\delta_j=\delta'_j$ for $j\neq i$, and 
$\delta_i\geq \delta_i'$. Similarly, we write 
$(\bfdelta,\bfeps)\leq_{\sye, i}(\bfdelta',\bfeps')$ when $\bfdelta=\bfdelta'$, 
$\varepsilon_j=\varepsilon'_j$ for $j\neq i$, and 
$\varepsilon_i\leq \varepsilon_i'$. 
Note that the calligraphic index $\syd$ indicates a
difference in a $\delta$ 
(and therefore, in an inequality of the corresponding
Gabrielov-Vorobjov system), 
while a calligraphic~$\sye$ does so for an $\varepsilon$ (and therefore in an equality). 
These relations  
capture the notion of a difference in only one entry of $\bfdelta$ or of $\bfeps$, 
respectively. The choice of the inequality in the $\varepsilon$s and the $\delta$s 
is different. This is done to ensure that if either $(\bfdelta,\bfeps)\leq_{\syd,i}(\bfdelta',\bfeps')$ or $(\bfdelta,\bfeps)\leq_{\sye,i}(\bfdelta',\bfeps')$, then we have, 
for all formulas $\Phi$ over $f$, the inclusion 
\begin{equation}\label{eq:incGV}
   \gv_{\bfdelta,\bfeps}(f,\Phi)\subseteq \gv_{\bfdelta',\bfeps'}(f,\Phi)
\end{equation}
between the corresponding Gabrielov-Vorobjov Approximations. 
We write 
\[
\left(\bfdelta,\bfeps\right)\subs{\syd,i}
\left(\bfdelta',\bfeps'\right)
\] 
to denote that 
$$
\left(\bfdelta,\bfeps\right)\leq_{\syd,i} 
\left(\bfdelta',\bfeps'\right)\quad \mbox{ or }\quad  
\left(\bfdelta',\bfeps'\right)\leq_{\syd,i} 
\left(\bfdelta,\bfeps\right).
$$
This notation is consistent with the meaning of updating 
$\left(\bfdelta,\bfeps\right)$ to 
$\left(\bfdelta',\bfeps'\right)$ by updating (either increasing or 
decreasing) only $\delta_i$ to $\delta_i'$. We similarly define 
$\left(\bfdelta,\bfeps\right)\subs{\sye,i} 
\left(\bfdelta',\bfeps'\right)$. 

The following result states the main property of these rewritings. 
 
\begin{prop}\label{thm:reductionGV1}
Let $f\in\Hd[q]$, $\Phi$ be a strict 
formula over $f$, and 
$\bfdelta, \bfdelta', \bfeps, \bfeps'\in\bbR^m$ be such that both $(\bfdelta,\bfeps)$ and $(\bfdelta',\bfeps')$ satisfy \eqref{GVconditionexplicit}. If either $(\bfdelta,\bfeps)\subs{\syd,i}(\bfdelta',\bfeps')$ or 
$(\bfdelta,\bfeps)\subs{\sye,i}(\bfdelta',\bfeps')$, then the corresponding 
inclusion~\eqref{eq:incGV} of Gabrielov-Vorobjov 
Approximations induces a homotopy equivalence.
\end{prop}

Proving Theorem \ref{thm:quantitativeGV} from this proposition is easy.

\begin{proof}[Proof of Theorem \ref{thm:quantitativeGV}]
By the definition of $\ll$, it is clear that there exist at least one $(\tilde{\bfdelta},\tilde{\bfeps})$ satisfying both~\eqref{GVcondition} 
and~\eqref{GVconditionexplicit}. 

For any $(\bfdelta,\bfeps)$ 
satisfying~\eqref{GVconditionexplicit}, we can easily construct 
a sequence $\left(\bfdelta^{(0)},\bfeps^{(0)}\right),\ldots,
\left(\bfdelta^{(\ell)},\bfeps^{(\ell)}\right)$ of pairs 
satisfying~\eqref{GVconditionexplicit} such that
\begin{enumerate}
\item $\left(\bfdelta^{(0)},\bfeps^{(0)}\right)=(\bfdelta,\bfeps)$,
\item $\left(\bfdelta^{(\ell)},\bfeps^{(\ell)}\right)=(\tilde{\bfdelta},\tilde{\bfeps})$, and
\item for each $p\in\{0,\ldots,\ell-1\}$, there are $k_p\in\{\syd,\sye\}$ and 
$i_p\in\{1,\ldots,m\}$ such that
$$\left(\bfdelta^{(p)},\bfeps^{(p)}\right)\subs{k_p,i_p} 
\left(\bfdelta^{(p+1)},\bfeps^{(p+1)}\right).
$$
\end{enumerate}
For such a sequence, the isomorphism types of the 
homology groups of 
$\gv_{\bfdelta^{(p+1)},\bfeps^{(p+1)}}(f,\Phi)$ 
don't change at each step as 
a consequence of Proposition~\ref{thm:reductionGV1}. 
Thus $\gv_{\bfdelta,\bfeps}(f,\Phi)$ has homology 
groups isomorphic to 
those of $\gv_{\tilde{\bfdelta},\tilde{\bfeps}}(f,\Phi)$. 
The conclusion now follows from 
applying Theorem~\ref{thm:GV} to the latter.
\end{proof}

We next focus on the situations 
$(\bfdelta,\bfeps)\subs{\syd,i}(\bfdelta',\bfeps')$ 
and $(\bfdelta,\bfeps)\subs{\sye,i}(\bfdelta',\bfeps')$. These situations correspond to replacing $\delta_i$ in 
the first one and $\varepsilon_i$ in the second one by some $\zeta\in (\varepsilon_i,\delta_i)$. Even though we are updating only one entry 
in the pair $(\bfdelta,\bfeps)$, we have to modify the inequalities associated to 
several polynomials. Instead of doing this replacement simultaneously in all the inequalities,
 we do it by steps, in the inequalities corresponding to a single polynomial at a time. 
 With this intuition at hand, we introduce the semialgebraic sets below.

Fix $f\in\Hd[q]$, a strict formula $\Phi$ over $f$, positive numbers 
$\delta,\varepsilon,\zeta,t$, 
and $a\in\{0,\ldots,q\}$. We define the following 
spherical semialgebraic sets:

$\gv_{\delta,\varepsilon,\zeta,t}^{\syd,a}(f,\Phi)$ is obtained 
from $\Phi$ by rewriting
\begin{align*}
  &f_l=0\rew |f_l(x)|\leq\varepsilon\|f_l\|\\
  &f_l>0 \rew \begin{cases} 
     f_l(x)\geq\delta\|f_l\| &\mbox{ if $l>a$}\\
     f_a(x)\geq t\|f_a\| & \mbox{ if $l=a$}\\  
     f_l(x)\geq \zeta\|f_l\| & \mbox{ if $l<a$}
     \end{cases}\\
  &f_l<0\rew \begin{cases} 
  f_l(x)\leq-\delta\|f_l\| & \mbox{ if $l>a$}\\
  f_a(x)\leq -t\|f_a\| &\mbox{ if $l=a$}\\
  f_l(x)\leq -\zeta\|f_l\|& \mbox{ if $l<a$.}
  \end{cases}
\end{align*}
$\gv_{\delta,\varepsilon,\zeta,t}^{\sye,a}(f,\Phi)$ is obtained from $\Phi$ by rewriting 
\begin{align*} 
  &f_l=0\rew \begin{cases} 
     |f_l(x)|\leq\varepsilon\|f_l\| &\mbox{ if $l>a$}\\
     |f_a(x)|\leq t\|f_a\| & \mbox{ if $l=a$}\\  
     |f_l(x)|\leq \zeta\|f_l\| & \mbox{ if $l<a$}
     \end{cases}\\
  &f_l>0 \rew f_l(x)\geq\delta\|f_l\|\\
  &f_l<0\rew f_l(x)\leq-\delta\|f_l\|.
\end{align*}
Consider now $\bfdelta,\bfeps\in(0,\infty)^m$, $\syc\in\{\syd,\sye\}$, 
$i\in\{1,\ldots,m\}$, $a\in [q]$
and $\zeta,t>0$. We define 
the intermediate Gabrielov-Vorobjov approximations as the sets
\begin{equation}\label{eq:defGVugly}
\gv_{\bfdelta,\bfeps,\zeta,t}^{\syc,i,a}(f,\Phi):=
\gv_{\delta_i,\varepsilon_i,\zeta,t}^{\syc,a}(f,\Phi) \quad \cup \quad
\bigcup_{j\neq i} \gv_{\delta_j,\varepsilon_j}(f,\Phi).
\end{equation}
In particular, we can see $\gv_{\bfdelta,\bfeps,\zeta,t}^{\syd,i,a}(f,\Phi)$ as the result of having 
replaced $\delta_i$ by $\zeta$ in all the 
inequalities with polynomials $f_1,\ldots,f_{a-1}$, and being in the process of 
making the replacement in those inequalities with $f_a$ with the parameter $t$ 
moving from $\delta_i$ to $\zeta$. 

We now observe that for $\zeta,t,t'>0$ with $t\le t'$ we have the inclusions
\begin{equation}\label{eq:incs}
   \gv_{\bfdelta,\bfeps,\zeta,t'}^{\syd,i,a}(f,\Phi)
   \subseteq 
   \gv_{\bfdelta,\bfeps,\zeta,t}^{\syd,i,a}(f,\Phi)
   \quad\text{ and }\quad
   \gv_{\bfdelta,\bfeps,\zeta,t'}^{\sye,i,a}(f,\Phi)
   \supseteq 
   \gv_{\bfdelta,\bfeps,\zeta,t}^{\sye,i,a}(f,\Phi). 
\end{equation}
The crucial fact needed to prove 
Theorem~\ref{thm:quantitativeGV} is that these 
inclusions induce homotopy equivalences. 

\begin{prop}\label{thm:reductionGV2}
Let $f\in\Hd[q]$, $\Phi$ be a strict formula, 
$\bfdelta,\bfeps\in\bbR^m$ satisfying~\eqref{GVconditionexplicit}, 
and let $i\in [m]$ and $a\in [q]$. Then, 
\begin{description}
\item[(1)] 
For all $\zeta\in(\varepsilon_i,\delta_i)$ and 
$\varepsilon_i<t\le t'<\varepsilon_{i+1}$ (where $\varepsilon_{m+1}=1/\sqrt{2}\,\kappabar(f)$ 
by convention), the inclusion
$\gv_{\bfdelta,\bfeps,\zeta,t'}^{\syd,i,a}(f,\Phi)
\subseteq\gv_{\bfdelta,\bfeps,\zeta,t}^{\syd,i,a}(f,\Phi)$
induces a homotopy equivalence. 
\item[(2)] 
For all $\zeta\in(\varepsilon_i,\delta_i)$ and 
$\delta_{i-1}<t\leq t'<\delta_{i}$ 
(where $\delta_0=0$ by convention), 
the inclusion
$\gv_{\bfdelta,\bfeps,\zeta,t'}^{\sye,i,a}(f,\Phi)
\supseteq\gv_{\bfdelta,\bfeps,\zeta,t}^{\sye,i,a}(f,\Phi)$
induces a homotopy equivalence. 
\end{description}
\end{prop}

Again, Proposition~\ref{thm:reductionGV1} easily follows from this result.

\begin{proof}[Proof of Proposition~\ref{thm:reductionGV1}]
Assume that 
$(\bfdelta,\bfeps)\leq_{\syd,i}(\bfdelta',\bfeps')$ holds. Then $\delta_i\ge \delta_i'$ and without loss of 
generality, $\delta_i>\delta_i'$. 
The following equalities then follow from the 
definition of 
$\gv_{\bfdelta,\bfeps,\zeta,t}^{1,i,a}(f,\Phi)$ 
(we omit the $(f,\Phi)$ in what follows for simplicity):
\begin{itemize}
\item $\gv_{\bfdelta,\bfeps,\delta'_i,\delta_i}^{\syd,i,1}
=\gv_{\bfdelta,\bfeps}$,
\item $\gv_{\bfdelta,\bfeps,\delta_i',\delta'_i}^{\syd,i,a}
=\gv^{\syd,i,a+1}_{\bfdelta,\bfeps,\delta'_i,\delta_i}$, 
\quad for all $a\in[q-1]$,
\item $\gv_{\bfdelta,\bfeps,\delta_i',\delta'_i}^{\syd,i,q}
=\gv_{\bfdelta',\bfeps}=\gv_{\bfdelta',\bfeps'}$,
\end{itemize}
the last equality as, by assumption, $\bfeps=\bfeps'$. 
These equalities 
yield the following chain
\[
 \gv_{\bfdelta,\bfeps}=\gv^{\syd,i,1}_{\bfdelta,\bfeps,\delta'_i,\delta_i}
 \subseteq \gv^{\syd,i,1}_{\bfdelta,\bfeps,\delta_i',\delta'_i}
 =\gv^{\syd,i,2}_{\bfdelta,\bfeps,\delta'_i,\delta_i}
 \subseteq \gv^{\syd,i,2}_{\bfdelta,\bfeps,\delta_i',\delta'_i} =
 \cdots
 \subseteq \gv^{\syd,i,q}_{\bfdelta,\bfeps,\delta_i',\delta'_i}
 =\gv_{\bfdelta',\bfeps'},
\]
in which all inclusions induce homotopy equivalences by 
Proposition~\ref{thm:reductionGV2}(1). 
Hence Proposition~\ref{thm:reductionGV1} follows in this case.

For the other cases, i.e.,~when $(\bfdelta',\bfeps')\leq_{\syd,i}(\bfdelta,\bfeps)$ 
or when $(\bfdelta,\bfeps)\subs{\sye,i}(\bfdelta',\bfeps')$, 
we proceed analogously. 
\end{proof}

\subsection{Some elements of differential topology}
\label{sec:difTop}

To prove Proposition~\ref{thm:reductionGV2} will require more elaborate arguments. 
We lay down these arguments in this section. To do so, we first review the basic 
notions of the Mather-Thom theory, already introduced in~\cite{BCTC1}, now using some more 
general results from~\cite{Gibson1976}. We then develop a generalization of the sign 
partition, called $(f,\bflambda)$-partition. We will show that these new partitions, 
under well-posedness assumptions on $f$, are Whitney stratifications. 

\subsubsection{Stratified sets and Mather-Thom theory}

We will take a more general approach than that 
in~\cite{BCTC1}, as this time we 
will consider Whitney stratifications for more general sets. For motivation, 
however, we refer to~\cite[\S4.1]{BCTC1}.

\begin{defi}\cite{Mather2012,Gibson1976}\label{defiwhitney} 
A \textit{Whitney stratification} of a subset $\Omega$ of a smooth manifold $\mcM$ 
of dimension~$m$ 
is a partition $\mcS$ of $\Omega$ into locally closed smooth submanifolds of~$\mcM$, 
called {\em strata}, such that:
\begin{description}
\item[F] (Locally finite)  Every $x\in \Omega$ has a neighborhood intersecting finitely 
many strata only.
\item[W] (Whitney's condition b) 
For every strata $\varsigma,\sigma\in\mcS$, 
every point $x\in\varsigma\cap \overline{\sigma}$, every sequence of points $\{x_{\ell}\}_{\ell\in\bbN}$ in 
$\varsigma$ converging to $x$, and every 
sequence of points 
$\{y_{\ell}\}_{\ell\in\bbN}$ in $\sigma$ converging to $x$,
we have that, in all local charts of $\mcM$ around $x$,  
\[
  \lim_{\ell\to\infty}\overline{x_{\ell},y_{\ell}}\subseteq
  \lim_{\ell\to\infty}\Tg_{y_\ell}\sigma,
\]
provided both limits exist. 
The inclusion should be interpreted in the local coordinates of the chart:
$\overline{x_{\ell},y_{\ell}}$ denotes the straight line joining $x_{\ell}$ and $y_{\ell}$, 
$\Tg_{y_{\ell}}\sigma$ denotes the affine plane 
tangent to $\sigma$ at $y_{\ell}$, 
and the limits are to be interpreted in the corresponding Grassmannians of~$\bbR^m$. 
\end{description}
\end{defi}

Recall that a map is {\em proper} when its inverse image of any compact subset is 
compact. Our interest on Whitney stratified sets is linked to the following result, a version 
of the so-called Thom's First Isotopy Lemma (see~\cite[Ch. II, (5.2)]{Gibson1976}). 

\begin{theo}\label{thm:ThomMatherRetraction}
Let $\mcM$ be a smooth manifold with a Whitney stratification $\mcS$ of a locally closed subset 
$\Omega\subseteq\mcM$ and let 
$\alpha:\mcM\rightarrow \bbR^k$ be a smooth map such that:
\begin{description}
\item{\rm (i)} $\alpha:\Omega\rightarrow \bbR^k$ is proper, 
\item{\rm (ii)} $\alpha_{|\sigma}:\sigma\rightarrow \bbR^k$ is a surjective submersion,
for each stratum $\sigma\in\mcS$.
\end{description}
Then $\alpha:\Omega\rightarrow\bbR^k$ is a trivial fiber bundle. That is, there exist a subset $F\subseteq \Omega$ 
and a homeomorphism 
$h=(h_{\bbR^k},h_F):\Omega\rightarrow \bbR^k\times F$ such that $\alpha=h_{\bbR^k}$.
Furthermore, for all $x,y\in\Omega$, 
$h_F(x)=h_F(y)$ implies that $x$ and $y$ lie in 
the same stratum of~$\mcS$.\eproof
\end{theo}

\begin{remark}
As the codomain of $\alpha$ is $\bbR^k$, it follows from the proof of~\cite[Ch. II, (5.2)]{Gibson1976} 
that we have a trivial fiber bundle and not just a locally trivial fiber bundle. 
The last sentence follows from noting that the trivial fibration in the statement of~\cite[Ch. II, (5.2)]{Gibson1976} is stratified, 
see~\cite[Ch. II, (5.1)]{Gibson1976}. 
\end{remark}

\subsubsection{\texorpdfstring{$(f,\bflambda)$}{(f,lambda)}-partitions}

A $q$-tuple $f\in\Hd[q]$ of polynomials leads to partition of the sphere $\bbS^{n}$ 
according to the $q$-tuples of signs obtained when evaluating the polynomials. 
We are going to make a finer classification by considering not only the signs, 
but also the magnitudes of the values with respect to some finite grid. 
This is the idea
of $(f,\bflambda)$-partitions.

\begin{defi}\label{def:frpartition}
Let $f\in\Hd[q]$ and $\bflambda\in\bbR^{q\times (m+1)}$ be a matrix whose entries satisfy 
$0=\lambda_{i,0}<\lambda_{i,1}<\cdots<\lambda_{i,m}$, for each $i\in [q]$. 
To each point $x\in\bbS^n$ we associate the following sets:
\begin{enumerate}
\item[(I$\Izero$)] For all $0\le k\le m$,  
$I_{\Izero,k}(x):= \{i\in [q]\mid |f_i(x)|/\|f_i\|=\lambda_{i,k}\}$.
\item[(I$\Ione$1)] For all $0\le k< m$, 
$I_{\Ione,k}(x):=\{i\in [q]\mid \lambda_{i,k}<|f_i(x)|/\|f_i\|<\lambda_{i,k+1}\}$.
\item[(I$\Ione$2)] $I_{\Ione,m}(x):=\{i\in [q]\mid \lambda_{i,m}<|f_i(x)|/\|f_i\|\}$.
\end{enumerate}
This defines the ordered partition of $[q]$ (in which we allow empty sets):
\[
\bfI(x):=(I_{\Izero,0}(x),I_{\Ione,0}(x),I_{\Izero,1}(x),I_{\Ione,1}(x),\ldots,I_{\Izero,m-1}(x),I_{\Ione,m-1}(x),I_{\Izero,m}(x),I_{\Ione,m}(x)) .
\]
The point $x$ also determines 
the tuple of sign conditions $\sigma(x)\in\{-1,0,+1\}^q$
given by 
\begin{enumerate}
\item[(S)] $\sigma_i(x):=\sgn(f_i(x))$\text{ for $i\in [q]$}.
\end{enumerate}
It is clear that 
$$
  x\sim y:= \big(\bfI(x)=\bfI(y) \textrm{ and } 
  \sigma(x)=\sigma(y)\big)
$$
describes an equivalence relation. 
We define 
the \textit{$(f,\bflambda)$-partition} $\Pe_{f,\bflambda}$ 
as the set of equivalence classes of this relation.

An ordered partition $\bfI:=(I_{\Izero,0},I_{\Ione,0},\ldots,I_{\Izero,m},I_{\Ione,m})$ of $[q]$ 
together with a sign vector $\sigma\in\{-1,0,+1\}^q$ defines the set
\[
 \pe_{\bfI,\sigma}:=\{x\in\bbS^n\mid \bfI(x)=\bfI,\,\sigma(x)=\sigma\},
\]
which is an element of $\Pe_{f,\bflambda}$, provided 
it is non-empty.
\end{defi}

\begin{remark}
Less formally, the construction of $\Pe_{f,\bflambda}$
can be described as follows: 
the $i$th  row of the matrix $\bflambda\in\bbR^{q\times (m+1)}$ 
defines a partition of the set $\bbR_{\ge 0}$ of nonnegative reals 
into the open intervals $(\lambda_{i,k},\lambda_{i,k+1})$
and the singleton sets $\{\lambda_{i,k}\}$.
By considering also the numbers $\pm\lambda_{i,k}$, 
we similarly obtain a partition of $\bbR$, which is symmetric 
with respect to the reflection $x\mapsto -x$. 
The product of these partitions of $\bbR$, for $i \in [q]$, yields 
the partition $\Pe_{f,\bflambda}$ of $\bbR^q$. 
So the sets of this partition are products of open intervals and singletons, 
describing where a value $(y_1,\ldots,y_q)\in\bbR^q$ is located within the 
discrete grid provided by the matrix $\bflambda$
\end{remark}

\begin{exam}
Figure~\ref{fig:estrat} shows, locally, an example on $\bbS^2$ with  
$q=2$, $m=2$ and $\lambda_{1,i}=\lambda_{2,i}=\lambda_i$. The thick curves correspond to the zero sets 
for $f_1$ and $f_2$. The dashed lines are level curves (for both $f_1$ and 
$f_2$) with levels $-\lambda_1$ and $\lambda_1$ and the dotted curves 
are the same for the levels $-\lambda_2$ and $\lambda_2$. All these curves partition the picture 
into~36 two-dimensional open regions, 60 open segments, and 25 points. Each of 
these 121 regions corresponds to an element in $\Pe_{f,\bflambda}$. We won't 
attempt to write down the details for each of them but in Table~\ref{table:1} 
we list some of them with their corresponding ordered partition. 

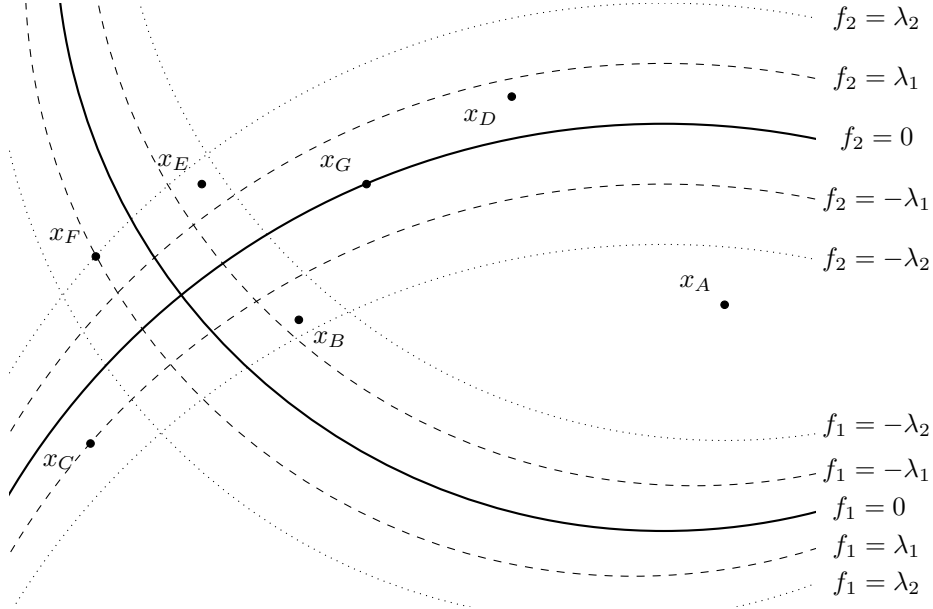
\begin{figure}[H]\centering
\begin{tikzpicture}[scale=2,point/.style={draw,minimum size=0pt,
inner sep=1pt,circle,fill=black}]
  \begin{scope}
    \clip (-0.3,0) rectangle (5,4);
   \begin{scope}
      \clip (4,-1.8) circle(5);
      \clip (4, -1.4) circle(5);
      \clip (4,4.5) circle(4);
    \end{scope}
    \draw[thick] (4,-1.8) circle(5);
    \draw[dashed] (4, -1.4) circle(5);
    \draw[dashed] (4, -2.2) circle(5);
    \draw[dotted] (4, -1) circle(5);
    \draw[dotted] (4, -2.6) circle(5);
    \draw[thick] (4,4.5) circle(4);
    \draw[dashed] (4.2,4.8) circle(4);
    \draw[dashed] (3.8,4.2) circle(4);
    \draw[dotted] (4.4,5.1) circle(4);
    \draw[dotted] (3.6,3.9) circle(4);
  \end{scope}
\path (5.4,3.9) node{\small $f_2=\lambda_2$};
\path (5.4,3.5) node{\small $f_2=\lambda_1$};
\path (5.4,3.1) node{\small $f_2=0$};
\path (5.4,2.7) node{\small $f_2=-\lambda_1$};
\path (5.4,2.3) node{\small $f_2=-\lambda_2$};
\path (5.4,1.2) node{\small $f_1=-\lambda_2$};
\path (5.4,0.9) node{\small $f_1=-\lambda_1$};
\path (5.35,0.65) node{\small $f_1=0$};
\path (5.4,0.4) node{\small $f_1=\lambda_1$};
\path (5.4,0.15) node{\small $f_1=\lambda_2$};
\draw (4.4,2) node(p) [point,label=120:{\small $x_A$}] {};
\draw (1.6,1.9) node(p) [point,label=335:{\small $x_B$}] {};
\draw (0.23,1.08) node(p) [point,label=185:{\small $x_C$}] {};
\draw (3,3.38) node(p) [point,label=220:{\small $x_D$}] {};
\draw (0.962,2.8) node(p) [point,label=100:{\small $x_E$}] {};
\draw (0.265,2.32) node(p) [point,label=140:{\small $x_F$}] {};
\draw (2.045,2.8) node(p) [point,label=140:{\small $x_G$}] {};
\end{tikzpicture}
\caption{An example (locally) on $\bbS^2$ with 
$S=\{1,2\}$ and $m=2$.}
\label{fig:estrat}
\end{figure}

\begin{table}[H]\centering
\begin{tabular}{|l|c|c|c|c|c|c|c|c|c|}
\hline\hline
 & $I_{\Izero,0}$ & $I_{\Ione,0}$ & $I_{\Izero,1}$  
 & $I_{\Ione,1}$ & $I_{\Izero,2}$ & $I_{\Ione,2}$ & $\sigma$ \\
\hline 
$x_A$ & $\varnothing$ & $\varnothing$ & $\varnothing$ 
& $\varnothing$ & $\varnothing$ &$\{1,2\}$&$(-1,-1)$\\
$x_B$ & $\varnothing$ & $\varnothing$  & $\varnothing$ 
& $\{1,2\}$ & $\varnothing$ &$\varnothing$&$(-1,-1)$\\
$x_C$ & $\varnothing$ & $\varnothing$  & $\varnothing$ 
& $\varnothing$  & $\{2\}$ &$\{1\}$&$(+1,-1)$\\
$x_D$ & $\varnothing$ & $\{2\}$ & $\varnothing$ 
& $\varnothing$  & $\varnothing$ &$\{1\}$&$(-1,+1)$\\
$x_E$ & $\varnothing$ & $\varnothing$ & $\varnothing$ 
& $\{1,2\}$  & $\varnothing$ &$\varnothing$&$(-1,+1)$\\
$x_F$ & $\varnothing$ & $\varnothing$ & $\{1\}$ 
& $\varnothing$ & $\{2\}$ &$\varnothing$&$(+1,+1)$\\
$x_G$ & $\{2\}$ & $\varnothing$ & $\varnothing$ 
& $\varnothing$  & $\varnothing$ &$\{1\}$&$(-1,~0)$\\
\hline\hline
\end{tabular}
\caption{Some points in Figure~\ref{fig:estrat} and their ordered 
partition.}\label{table:1}
\end{table}
\end{exam}
\bigskip 

The following theorem gives sufficient conditions on
$f\in\Hd[q]$ 
and $\bflambda\in\bbR^{q\times (m+1)}$ for the $(f,\bflambda)$-partition 
of $\bbS^n$ to be a Whitney stratification. 

\begin{theo}\label{thm:flambdapartition}
Let $f\in\Hd[q]$ with $\kappabar(f)<\infty$ and assume $\bflambda\in\bbR^{q\times (m+1)}$ satisfies for $i\in [q]$,
\begin{equation}\label{condititionFL}
  0=\lambda_{i,0}<\lambda_{i,1}<\cdots<\lambda_{i,m}<\frac{1}{\sqrt{2}\,\kappabar(f)}.
\end{equation}
Then the $(f,\bflambda)$-partition $\Pe_{f,\bflambda}$ 
is a Whitney stratification of $\bbS^n$.
Furthermore, under these conditions, the following holds:
\begin{enumerate}
\item[(1)] The codimension in $\bbS^n$ of each stratum $\pe_{\bfI,\sigma}$ equals $\sum_{k=0}^m |I_{\Izero,k}|=q-\sum_{k=0}^m|I_{\Ione,k}|$.
\item[(2)] Given $\pe_{\bfI,\sigma}\in\Pe_{f,\bflambda}$ and $a\in I_{\Ione,k}$ for some $k<m$, the map
\begin{align*}
\hat{f}_{a,\bfI,\sigma}:\pe_{\bfI,\sigma}&\rightarrow (\lambda_{a,k},\lambda_{a,k+1})\\
x&\mapsto |f_{a}(x)|/\|f_{a}\|
\end{align*}
is a surjective submersion.
\end{enumerate}
\end{theo}

\begin{remark} 
Recall that the condition ``$a\in I_{\Ione,k}$ for some $k<m$" can  be less cryptically written 
as ``$|f_a(x)|/\|f_a\|\in (\lambda_{a,k},\lambda_{a,k+1})$ for some
$x\in \pe_{\bfI,\sigma}$ and $k<m$", 
or simply as ``$|f_a|/\|f_a\|\in (\lambda_{a,k},\lambda_{a,k+1})$ on $\pe_{\bfI,\sigma}$ for some $k<m$". 
\end{remark}


\begin{proof}
In order to show that $\Pe_{f,\bflambda}$ is a Whitney stratification, we notice that
\[
  \Pe_{f,\bflambda}=\bigwedge_{i=1}^q\Pe_{f_i,\lambda_i}
  :=\left\{\cap_{i=1}^q\pe_i\mid\pe_i\in
  \Pe_{f_i,\lambda_i}\right\}
\]
where $\lambda_i:=(\lambda_{i,0},\ldots,\lambda_{i,m})$ 
is the $i$th row of $\bflambda$ and 
$\Pe_{f_i,\lambda_i}$ is the 
$(f_i,\lambda_i)$-partition of $\bbS^n$. 
Thus, in virtue of~\cite[Ch.~I, (1.3)]{Gibson1976}, it is enough to show that each $\Pe_{f_i,\lambda_i}$ is 
a Whitney stratification and that 
$\Pe_{f_1,\lambda_1},\ldots,\Pe_{f_q,\lambda_q}$ 
are \emph{in general position}, which, 
following~\cite[Ch.~I, (1.2)]{Gibson1976}, means that 
for all $\pe_1\in\Pe_{f_1,\lambda_1},\ldots,
\pe_q\in\Pe_{f_q,\lambda_q}$,
\[
  \mathrm{codim}\,\bigcap_{i=1}^q\Tg_x\pe_i
  =\sum_{i=1}^q\mathrm{codim}\,\Tg_x\pe_i\text{,}
\]
holds for each $x\in \bigcap_{i=1}^q\pe_i$, where
$\mathrm{codim}$ stands for codimension.

Note that $\Pe_{f_i,\lambda_i}$ consists of open sets of the form $f_i^{-1}(a,b)$, 
with $(a,b)$ an open interval, or a hypersurface of 
the form $f^{-1}_i(t)$, with $t=\|f_i\|\lambda_{i,j}$ 
for some $j$. By assumption on $\bflambda$, this implies 
that for such $t$, 
$|t|<\|f_i\|/(\sqrt{2}\kappabar(f))$ 
and hence, by~\cite[Proposition~3.6]{BCTC1} and the 
implicit function theorem, all the hypersurfaces 
are smooth. Whitney's
condition~b is verified in a straighforward way so that we conclude that $\Pe_{f_i,\lambda_i}$ is a 
Whitney stratification.

We show now that $\Pe_{f_1,\lambda_1},\ldots,\Pe_{f_q,\lambda_q}$ are in general position. Let $\pe_i\in\Pe_{f_i,\lambda_i}$, for $i\in[q]$, and 
$x\in\cap_{i\leq q}\pe_i$. It is easy to check that 
$\mathrm{codim}\,\Tg_x\pe_i=1$ if $i\in I_{\Izero,k}(x)$ and 
$\mathrm{codim}\,\Tg_x\pe_i=0$ otherwise. Therefore, abbreviating $I_{\Izero,\ast}(x):=\bigcup_kI_{\Izero,k}(x)$, we get 
$\sum_{i=1}^q\mathrm{codim}\,\Tg_x\pe_i =|I_{\Izero,\ast}(x)|$.
In addition, when $\pe_i$ is a hypersurface, we have $\Tg_x\pe_i=\ker\diff_xf_i$, thus
\[
  \bigcap_{i=1}^q\Tg_x\pe_i=\bigcap_{i\in I_{\ast,0}(x)}\ker\diff_xf_i=
  \ker\diff_xf^{I_{\Izero,\ast}(x)}.
\]
By~\cite[Proposition~3.6]{BCTC1}, the codimension of 
the right-hand side is $|I_{\Izero,\ast}(x)|$. This shows 
that $\Pe_{f_1,\lambda_1},$ $\ldots,$ 
$\Pe_{f_q,\lambda_q}$ are in general position. 
We conclude that $\Pe_{f,\bflambda}$ is a Whitney 
stratification.

The argument above proves also~(1).

We prove part~(2) is a standard way. First,
we will show that $\hat{f}_{a,\bfI,\sigma}$ is a 
submersion, i.e., that its gradient is not tangent to
$\pe_{\bfI,\sigma}$. Then we will show that
$\hat{f}_{a,\bfI,\sigma}$ is closed. Assume we have 
proved these claims. 
The fact that $\hat{f}_{a,\bfI,\sigma}$ is a submersion 
implies that it is open. Therefore, as 
$\hat{f}_{a,\bfI,\sigma}$ is also closed 
and $(\lambda_{i,k},\lambda_{i,k+1})$ is connected, we deduce that $\hat{f}_{a,\bfI,\sigma}$ is surjective. 
This will finish the proof.

To show that $\hat{f}_{a,\bfI,\sigma}$ is a submersion, we fix a point 
$p\in\pe_{\bfI,\sigma}$ and take trivializing coordinates around it, 
using~\cite[Lemma~4.24]{BCTC1}. In these coordinates, using the notation from~\cite{BCTC1}, $\pe_{\bfI,\sigma}$ 
is an open subset of an affine subspace given by
\begin{equation}\label{eq:systempebar}
\begin{cases}
U_i=\sigma_i\lambda_{i,k}&(0<k\leq m,\,i\in I_{\Izero,k})\\
\lambda_{i,k}\leq\sigma_iU_i\leq\lambda_{i,k+1}&(0\leq k<m,\,l\geq 1,\,i\in I_{\Ione,k})
\end{cases}
\end{equation}
whose tangent space is given by the system
\begin{equation}\label{eq:systempetangent}
\begin{cases}
U_i=0&(0\leq k\leq m,\,i\in I_{\Izero,k})
\end{cases}.
\end{equation}
The map $\hat{f}_{a,\bfI,\sigma}$ in these 
coordinates becomes the linear map $U_{a}$. 
To check that $\hat{f}_{a,\bfI,\sigma}$ is a 
submersion is then enough to check that $U_a$ 
is not identically zero in the tangent space 
in these coordinates. Since 
$a\notin \cup_{k}I_{\Izero,k}$, $U_{a}$, this is the 
case and so $\hat{f}_{a,\bfI,\sigma}$ is a submersion.

To show that $\hat{f}_{a,\bfI,\sigma}$ is closed, it 
is enough to show that for every sequence 
$\{x_k\}$ in $\pe_{\bfI,\sigma}$, if 
$\{\hat{f}_{a,\bfI,\sigma}(x_k)\}$ has a limit 
$\lambda\in (\lambda_{a,k},\lambda_{a,k+1})$, then 
there exists $x\in \pe_{\bfI,\sigma}$ such that 
$\hat{f}_{a,\bfI,\sigma}(x)=\lambda$.

As $\bbS^n$ is compact, we can assume without loss of generality that $\{x_k\}$ converges 
to a point $x'\in \overline{\pe_{\bfI,\sigma}}$. 
By continuity, $\hat{f}_{a,\bfI,\sigma}(x')=\lambda$. 
Passing again to trivializing coordinates 
and using~\cite[Lemma~4.24]{BCTC1}, we perturb $x'$ to a point $x$ whose components in these coordinates are 
as follows:
\[
u_i:=\begin{cases}
u'_i+\sigma_it&\text{if for some $k$ we have } 
i\in I_{k,1}\text{ and }u'_i=\lambda_{i,k}\\
u'_i-\sigma_it&\text{if for some $k$ we have } 
i\in I_{k,1}\text{ and }u'_i=\lambda_{i,k+1}\\
u'_i&\text{otherwise}
\end{cases}
\]
with a sufficiently small $t>0$. This new point 
$x$ evaluates to the same value as $x'$ 
under $\hat{f}_{a,\bfI,\sigma}$, since $u_{a}=u'_a$ as 
$u'_a=\lambda\in(\lambda_{a,k},\lambda_{a,k+1})$ 
by hypothesis; and it belongs to 
$\pe_{\bfI,\sigma}$. Thus it is the desired point and we are done.
\end{proof}

\subsection{Proof of Proposition~\ref{thm:reductionGV2}}
\label{sec:proof_red}

We have now all the tools needed to prove Proposition~\ref{thm:reductionGV2} and with it 
to finish 
the proof of the Quantitative Gabrielov-Vorobjov Theorem~\ref{thm:quantitativeGV}. 
We will only prove part (1) of Proposition~\ref{thm:reductionGV2} as part (2) is proven in an analogous way.

We fix $f\in\Hd[q]$, a strict formula $\Phi$ over $f$, 
tuples $\bfdelta,\bfeps\in(0,\infty)^m$, an index 
$i\in[m]$, 
a point $\zeta\in(\varepsilon_i,\delta_i)$, points 
$t<t'$ in the interval $(\varepsilon_i,\varepsilon_{i+1})$,  
and an index $a\in[q]$, as in the statement 
of Proposition~\ref{thm:reductionGV2} and satisfying the hypothesis given there. 
Since $a$ is fixed, we can assume $\|f_a\|=1$ without loss of generality after scaling $f$ appropriately. 

We also choose positive numbers $t_0,t_1$ satisfying 
$$
 \varepsilon_i < t_0 < t < t' < t_1 < \varepsilon_{i+1} 
$$
and define the matrix $\bflambda\in\bbR^{q\times(2m+2)}$ whose 
$l$th row $\lambda_{l}$ is given by  
\begin{equation}\label{eq:defLambda}
 \lambda_{l}:=\begin{cases}
 (0,\varepsilon_1,\delta_1,\ldots,\varepsilon_i,\zeta,
 \delta_i,\varepsilon_{i+1}\ldots,\varepsilon_m,
 \delta_m), &\text{if }l\neq a,\\
 (0,\varepsilon_1,\delta_1,\ldots,\varepsilon_i,t_0,
 t_1,\varepsilon_{i+1},\ldots,\varepsilon_m,\delta_m),
 &\text{if }l=a.
\end{cases}
\end{equation}
By construction, this $\bflambda$ satisfies \eqref{condititionFL}. 
We will assume these conventions throughout this 
subsection without further mentioning them explicitly.
The matrix $\bflambda$ determines the $(f,\bflambda)$-partition $\Pi_{f,\bflambda}$ which, 
as we saw in Theorem~\ref{thm:flambdapartition}, is a Whitney stratification of $\bbS^n$. 

Recall the intermediate Gabrielov-Vorobjov approximations 
$$
  \gv_\tau := \gv^{\syd,i,a}_{\bfdelta,\bfeps,\zeta,\tau}(f,\Phi),
$$ 
defined in \eqref{eq:defGVugly}, for $\tau\in[t_0,t_1)$. 
These are compact subsets of $\bbS^n$.

Proposition~\ref{thm:reductionGV2} claims that 
$\iota:\gv_{t'}\rightarrow \gv_{t}$ is a homotopy equivalence. 
The basic idea for showing this is to apply Theorem~\ref{thm:ThomMatherRetraction} 
to the stratification provided by $\Pe_{f,\bflambda}$.
In a first step towards this goal, we describe 
how the strata $\pe_{\bfI,\sigma}$ of $\Pi_{f,\bflambda}$ intersect~$\gv_\tau$.
The findings are summarized in the proposition below, whose easy  
but somewhat cumbersome proof is postponed to~\S\ref{se:TS}.

\begin{prop}\label{prop:strataunion}
\begin{description}
\item[(1)] $\gv_{t_0}$ is a union of strata of 
$\Pe_{f,\bflambda}$.
\item[(2)] Let $\tau,\tau'\in(t_0,t_1)$. For each $\pe_{\bfI,\sigma}\in\Pe_{f,\bflambda}$ such that $\pe_{\bfI,\sigma}\subseteq \gv_{t_{0}}$, the following holds:
\begin{description}
\item[(i)] $\pe_{\bfI,\sigma}\cap \gv_{\tau}=\varnothing$ 
if and only if $\pe_{\bfI,\sigma}\cap \gv_{\tau'}=\varnothing$. In this case, $\pe_{\bfI,\sigma}\subseteq |f_a|^{-1}(t_0)$.
\item[(ii)] $\pe_{\bfI,\sigma}\cap \gv_{\tau}=\pe_{\bfI,\sigma}$ if 
and only if 
$\pe_{\bfI,\sigma}\cap \gv_{\tau'}=\pe_{\bfI,\sigma}$.
\item[(iii)] If $\varnothing\neq\pe_{\bfI,\sigma}\cap 
\gv_{\tau}\subsetneqq\pe_{\bfI,\sigma}$, then 
$\pe_{\bfI,\sigma}\subseteq |f_a|^{-1}(t_0,t_1)$ and
\[
  \pe_{\bfI,\sigma}\cap \gv_\tau=\pe_{\bfI,\sigma} 
  \cap\{x\in\bbS^n\mid |f_a(x)|\geq \tau\}.
\]
\end{description}
\end{description}
\end{prop}






\subsubsection{Homotopies preserving \texorpdfstring{$\Pe_{f,\bflambda}$}{Pe(f,lambda)}}


We are now going to construct the maps and homotopies to show that the inclusion 
$\iota:\gv_{t'}\rightarrow \gv_{t}$ is a homotopy equivalence. 
For this, we should construct a continuous map $\rho:\gv_{t}\rightarrow \gv_{t'}$ 
and homotopies between the compositions of these maps and the identity maps.

A first approach would be to move around the points of $\gv_t\setminus\gv_{t'}$ and 
then extend the maps obtained continuously to the whole space. 
It is easier though to work
in the larger space $\gv_{t_0}\cap |f_a|^{-1}(t_0,t_1)$, 
where we can control what happens at the boundary and thus obtain the continuous extensions.

Consider the open subset 
$\mcM := \bbS^n \setminus f_a^{-1}(0)$ of $\bbS^n$ together with the smooth map 
$\mcM\to\bbR,\, s\mapsto |f_a(x)|$, 
as well as the locally closed set 
$$
 \Omega:=\gv_{t_0}\cap |f_a|^{-1}(t_0,t_1) \subseteq \mcM .
$$
By Proposition~\ref{prop:strataunion}(1), 
$\Omega$ is the union of certain strata $\pe_{\bfI,\sigma}$ of $\Pe_{f,\bflambda}$,  
namely of those strata on which $|f_a|$ takes values in $(t_0,t_1)$. 
We note that the restriction of $|f_a|$,  
\[
\alpha:\Omega\rightarrow (t_0,t_1),\,
x\mapsto |f_a(x)| ,
\]
is a proper map. Indeed,  the inverse image 
$\alpha^{-1}(J) = \{x\in \gv_{t_0} \mid f_a(x) \in J \}$ 
of a compact subset $J\subseteq (t_0,t_1)$ is a closed subset of 
the compact set $\gv_{t_0}$ and thus compact itself.

By Theorem~\ref{thm:flambdapartition}, $\Pe_{f,\bflambda}$
restricts to a Whitney stratification of $\Omega$ 
and the map $\alpha$ satisfies the hypothesis of 
Theorem~\ref{thm:ThomMatherRetraction}. 
Therefore, there is a subset $F\subseteq \Omega$  
and a homeomorphism 
$h:\Omega\rightarrow F\times (t_0,t_1)$ such that 
the following diagram commutes
\[
  \begin{tikzcd}
  \Omega \arrow[r,"\quad h\quad"] \arrow[dr,"\alpha"']
  & F\times (t_0,t_1) \arrow[d,"\pi_{(t_0,t_1)}"]\\
  & (t_0,t_1).&
\end{tikzcd}
\]
Moreover, the stratum in which $x\in\Omega$ 
lies only depends on $h_F(x)$, that is, 
if $h_F(x)=h_F(y)$ then $x$ and $y$ belong to the 
same stratum of $\Pe_{f,\bflambda}$. 

Consider the following continuous (piecewise linear) map 
\begin{eqnarray*}
\upsilon:[0,1]\times[t_0,t_1]&\rightarrow& [t_0,t_1]\\
(s,y)&\mapsto&
\begin{cases}
y&\text{if }y\in [t',t_1],\\
(1-s)y+st'&\text{if }y\in [t,t'],\\
\frac{t_0+t}{2}+\left((1-s)+s\frac{2t'-t-t_0}{t-t_0}\right)\left(y-\frac{t_0+t}{2}\right)
&\text{if }y\in [(t_0+t)/2,t],\\
y&\text{if }y\in [t_0,(t_0+t)/2].
\end{cases}
\end{eqnarray*}
One easily verifies that this map restricts to a continuous retraction of $[t,t_1]$ onto $[t',t_1]$ 
that leaves fixed all points in a neighborhood of $\{t_0,t_1\}$. 
With the help of~$\upsilon$, one defines the continuous map
\begin{align*}
\psi:[0,1]\times \Omega&\rightarrow\Omega\\
(s,x)&\mapsto \begin{cases}
x,&\text{if }\alpha(x)\notin ((t_0+t)/2,(t_1+t')/2),\\
h^{-1}(h_F(x),\upsilon(s,\alpha(x))),&\text{otherwise.}
\end{cases}
\end{align*}
The properties of $\upsilon$ and $h$ imply that this map restricts to a
continuous retraction of $\alpha^{-1}[t,t_1)$ onto $\alpha^{-1}[t',t_1)$ 
that leaves fixed all points in a neighborhood of the boundary 
$\partial\Omega\cap \gv_{t_0}$
of $\Omega$ in $\gv_{t_0}$ 
(note $\partial\Omega 
\subseteq |f_a|^{-1}(\{t_0,t_1\})$).

We also have that $\psi(s,\pe_{\bfI,\sigma})\subseteq \pe_{\bfI,\sigma}$ for all $s\in[0,1]$,
provided $\pe_{\bfI,\sigma}\subseteq\Omega$.
This is so, because the value $h_F(x)$ determines the 
stratum to which $x$ belongs and moreover 
$h_F(\psi(s,x))=h_F(x)$.

Since $\psi$ fixes all points in a neighborhood of 
$\partial\Omega\cap \gv_{t_0}$, it can be 
extended to the continuous map 
$\Psi:[0,1]\times \gv_{t_0}
\rightarrow \gv_{t_0}$ 
given by
\[
   \Psi(s,x)=\begin{cases}
   \psi(s,x),&\text{if }x\in \Omega,\\
   x,&\text{otherwise}.
   \end{cases}
\]
As we are extending by the identity, all properties of $\psi$ are inherited by $\Psi$. In other words, $\Psi$ restricts to a
continuous retraction of $\gv_{t_0}\cap |f_a|^{-1}[t,\infty)=(\gv_{t_0}\setminus \Omega)\cup \alpha^{-1}[t,t_1)$ onto 
$\gv_{t_0}\cap|f_a|^{-1}[t',\infty)$ and it preserves the stratification $\Pe_{f,\bflambda}$, i.e., 
we have
$\Psi(s,\pe_{\bfI,\sigma})\subseteq \pe_{\bfI,\sigma}$,
for all $\pe_{\bfI,\sigma}\in\Pe_{f,\bflambda}$ contained in $\gv_{t_0}$ and all 
$s\in[0,1]$.

We are now ready to conclude. 
However, as a warning, we note that $\Psi$ does not give a continuous retraction of $\gv_t$ onto $\gv_{t'}$. 
The reason is that $\gv_{\tau}=\gv_{t_0}\cap |f_a|^{-1}[\tau,\infty)$ generally does not hold!

\begin{proof}[Proof of Proposition~\ref{thm:reductionGV2}]
We first show that for all $s\in[0,1]$,
$$
\Psi(s,\pe_{\bfI,\sigma} \cap \gv_t) \subseteq \pe_{\bfI,\sigma} \cap \gv_t 
\quad\mbox{and}\quad 
\Psi(s,\pe_{\bfI,\sigma} \cap \gv_{t'} ) \subseteq \pe_{\bfI,\sigma} \cap \gv_{t'} .
$$ 

By Proposition~\ref{prop:strataunion}(2), there are three possible cases 
for each of these intersections. We only focus on the third one, \textbf{(iii)}, 
since the other two cases are straightforward.  
In this case, we have 
$\pe_{\bfI,\sigma}\cap \gv_t =\pe_{\bfI,\sigma}\cap \{x\mid |f_a(x)|\geq t\}$ and 
$|f_a|(\pe_{\bfI,\sigma})\subseteq (t_0,t_1)$. 
Thus $\pe_{\bfI,\sigma}\subseteq \Omega$ and
\[
\pe_{\bfI,\sigma}\cap \gv_t=\pe_{\bfI,\sigma}\cap|f_a|^{-1}[t,\infty).
\]
Since this is the case, again by Proposition~\ref{prop:strataunion}(2), the same happens for $t'$ and so
\[
\pe_{\bfI,\sigma}\cap \gv_{t'}=\pe_{\bfI,\sigma}\cap|f_a|^{-1}[t',\infty).
\]
Since $\Psi$ gives a deformation retract of $\gv_{t_0}\cap\alpha^{-1}[t,t_1)$ onto $\gv_{t_0}\cap\alpha^{-1}[t',t_1)$,  
it preserves the stratification $\Pe_{f,\bflambda}$, 
and moreover $\Psi$ gives a continuous retraction of 
$\pe_{\bfI,\sigma}\cap|f_a|^{-1}[t,\infty)=\pe_{\bfI,\sigma}\cap(\gv_{t_0}\cap |f_a|^{-1}[t,\infty))$ onto $\pe_{\bfI,\sigma}\cap|f_a|^{-1}[t',\infty)$. 
Hence $\Psi$ must preserve $\pe_{\bfI,\sigma}\cap \gv_t$ and $\pe_{\bfI,\sigma}\cap \gv_{t'}$ 
and we have shown the claim. 

We conclude that 
$\Psi(s,\gv_{t})\subseteq \gv_{t}$ and $\Psi(s,\gv_{t'})\subseteq \gv_{t'}$ for all $s\in[0,1]$.

This allows us 
to restrict $\Psi$ to obtain continuous maps
$$
\Theta:[0,1]\times \gv_t \rightarrow \gv_t,\ (s,x) \mapsto \Psi(s,x)
$$
and
$$
\Theta':[0,1]\times \gv_{t'} \rightarrow \gv_{t'},\ (s,x) \mapsto \Psi(s,x).
$$
Let
$\rho:\gv_t\rightarrow 
\gv_{t'}$ be the continuous surjection given by
\[
   x\mapsto \Psi(1,x).
\]
By examining the three cases of Proposition~\ref{prop:strataunion}(2), we see 
that $\rho$ is well-defined. 
Recall that $\iota:\gv_{t'}\rightarrow \gv_{t}$ is the inclusion map.  
By construction, we have  
\[
   \Theta_0=\mathrm{id}_{\gv_t}\text{, } \quad
   \Theta_1=\rho=\iota\circ \rho\text{, }\quad
   \Theta'_0=\mathrm{id}_{\gv_{t'}}\quad\text{ and }
   \quad
   \Theta'_1=\rho\circ\iota.
\]
Hence, both $(\mathrm{id}_{\gv_t},\iota\circ \rho)$ and 
$(\mathrm{id}_{\gv_{t'}},\rho\circ\iota)$ are pairs of  
homotopic maps. Thus $\iota$ induces an homotopy equivalence as desired.
\end{proof}

\subsubsection{Proof of Proposition~\ref{prop:strataunion}}
\label{se:TS}

The way we prove this proposition is by reducing to the basic case. To do this, 
we write~$\Phi$ in the form
\[
  \Phi\equiv \bigvee_{\xi\in\Xi}\phi_\xi
\]
where $\Xi$ is a finite set
and each $\phi_\xi$ is saturated. 
By a \textit{saturated} formula over $f\in\Hd[q]$ we mean a purely conjunctive strict formula over $f$ 
of the form
\[\bigwedge_{i=1}^q(f_i\propto 0)\]
where $\propto\in\{<,=,>\}$, i.e., a purely conjunctive strict formula over $f$ in which 
all components of $f$ occur. 
We use the word 
`saturated' to emphasize that there are no 
polynomials 
left in $f$ to add to the formula.

This is possible by writing $\Phi$ in DNF 
and then splitting the clauses were some polynomial is missing by adding these 
missing polynomials with all possible sign conditions. This rewriting  
does not alter the sets under consideration and hence, neither does so with the 
validity of the statement we want to prove.

As we can take out unions in \eqref{eq:defGVugly}, we have 
\begin{equation}\label{eq:uniongvsets}
\gv_\tau=\gv^{\syd,i,a}_{\bfdelta,\bfeps,\zeta,\tau}(f,\Phi) 
=\bigcup_{\xi\in\Xi}\left(\gv^{\syd,a}_{\delta_i,\varepsilon_i,\zeta,\tau}(f,\phi_\xi)\cup\bigcup_{j\neq i}\gv_{\delta_j,\varepsilon_j}(f,\phi_\xi)\right).
\end{equation}
Hence it is enough to consider how the different strata intersect with the sets in the right hand side. 
This is done in Lemmas~\ref{lem:intersectionA},~\ref{lem:intersectionB}, 
and~\ref{lem:intersectionC} below. We recall that we assume $\|f_a\|=1$ without loss of generality.

Before enunciating the lemmas, we associate to each saturated 
formula $\phi$
a \emph{sign vector} $\sgn(\phi)\in\{-1,0,1\}^q$   
given by $\sgn_i(\phi)$ equal to $-1$, $0$ or 
$1$ depending on whether $\propto_i$ is $<$, $=$ 
or $>$, respectively. It is clear 
that for any $x\in\bbS^n$, $x\in \Ap(f,\phi)$ if and only if $\sgn(f(x))=\sgn(\phi)$. 
We endow $\{-1,0,1\}^q$ with a partial order: 
we say $\sigma \preccurlyeq \sigma'$ iff for all $i$,
$\sigma_i \ne 0$ implies $\sigma_i =\sigma'_i$.   

The first lemma deals with the $\gv$ blocks of the form $\gv_{\delta_j,\varepsilon_j}(f,\phi_\xi)$ with $j\neq i$, 
the second lemma with those of the form $\gv^{\syd,a}_{\delta_i,\varepsilon_i,\zeta,t_0}$, 
and the third lemma with those of the form $\gv^{\syd,a}_{\delta_i,\varepsilon_i,\zeta,\tau}$ with $\tau\in(t_0,t_1)$. 
Of these, the third lemma is the most delicate one, 
as in this case, the $\gv$ blocks do not decompose as a union of strata.

\begin{lem}\label{lem:intersectionA} 
Let $\phi$ be a saturated formula over $f$, 
let $j\ne i$ and put
$\delta:=\delta_j$,  
$\varepsilon:=\varepsilon_j$.
For every 
$\pe_{\bfI,\sigma}\in\Pe_{f,\bflambda}$ the following are equivalent:
\begin{enumerate}
\item[(0I1)] $\pe_{\bfI,\sigma}\cap \gv_{\delta,\varepsilon}(f,\phi)\neq \varnothing$.
\item[(0I2)] $\pe_{\bfI,\sigma}\subseteq \gv_{\delta,\varepsilon}(f,\phi)$.
\item[(0I3)] $\sgn(\phi)\preccurlyeq \sigma$ 
and for all $l\in [q]$,
\begin{empheq}[left = \empheqlbrace]{align*}
|f_l|/\|f_l\| \leq\varepsilon\text{ on }\pe_{\bfI,\sigma},
&\text{\quad if }\sgn_l(\phi)=0,\\
|f_l|/\|f_l\| \geq\delta\text{ on }\pe_{\bfI,\sigma},
&\text{\quad if }\sgn_l(\phi)\neq 0.
\end{empheq}
\end{enumerate}
\end{lem}

\begin{proof}
The chain of implications from (0I3) to (0I2) to (0I1) follows directly from the 
definition of $\pe_{\bfI,\sigma}$. Therefore we only show that (0I1) implies (0I3).

Let $x\in \pe_{\bfI,\sigma}\cap\gv_{\delta,\varepsilon}(f,\phi)$. For each 
$l\in [q]$, we distinguish three cases:
\begin{enumerate}

\item[+)]If $\sgn_l(\phi)=1$, then $x\in \gv_{\delta,\varepsilon}(f,\phi)$ implies
$f_l(x)/\|f_l\| \geq \delta$. Therefore, $\sigma_l=\sgn(f_l(x))=1\succcurlyeq 1=\sgn_l(\phi)$ and 
$|f_l|/\|f_l\| \geq \delta$ on $\pe_{\bfI,\sigma}$. The latter because $\delta$ appears in in $\lambda$,  
and so either all $x\in\pe_{\bfI,\sigma}$ satisfy $|f_l(x)|/\|f_l\|\geq \delta$ or none of them does.

\item[$-$)]If $\sgn_l(\phi)=-1$, the argument is analogous to that of the case $\sgn_l(\phi)=1$.

\item[0)] If $\sgn_l(\phi)=0$, then 
$x\in \gv_{\delta,\varepsilon}(f,\phi)$ implies 
$|f_l(x)|/\|f_l\|\leq \varepsilon$. This, in turn, 
implies $|f_l|/\|f_l\|\leq \varepsilon$ on $\pe_{\bfI,\sigma}$, since $\varepsilon$ appears in $\lambda$,  
and so either all $x\in \pe_{\bfI,\sigma}$ satisfy this or none does. 
Also $0\preccurlyeq 0,+1,-1$, and so $\sgn(f_l(x))\preccurlyeq \sigma_l$.
\end{enumerate}
\end{proof}

\begin{lem}\label{lem:intersectionB} 
Let $\phi$ be a saturated formula over $f$. For every 
$\pe_{\bfI,\sigma}\in\Pe_{f,\bflambda}$, the following are equivalent:
\begin{enumerate}
\item[(1I1)] $\pe_{\bfI,\sigma}\cap \gv_{\delta_i,\varepsilon_i,\zeta,t_0}^{\syd,a}(f,\phi)\neq \varnothing$.
\item[(1I2)] $\pe_{\bfI,\sigma}\subseteq \gv_{\delta_i,\varepsilon_i,\zeta,t_0}^{\syd,a}(f,\phi)$.
\item[(1I3)] $\sgn(\phi)\preccurlyeq \sigma$ and, for all $l\in [q]$,
\begin{empheq}[left = \empheqlbrace]{align*}
|f_l|/\|f_l\| \leq\varepsilon_i\text{ on }\pe_{\bfI,\sigma},
&\text{\quad if }\sgn_l(\phi)=0,\\
|f_l|/\|f_l\| \geq\delta_i\text{ on }\pe_{\bfI,\sigma},
&\text{\quad if }\sgn_l(\phi)\neq 0\text{ and }l>a,\\
|f_l|/\|f_l\| \geq t_0\text{ on }\pe_{\bfI,\sigma},
&\text{\quad if }\sgn_l(\phi)\neq 0\text{ and }l=a,\\
|f_l|/\|f_l\| \geq\zeta\text{ on }\pe_{\bfI,\sigma},
&\text{\quad if }\sgn_l(\phi)\neq 0\text{ and }l<a.
\end{empheq}
\end{enumerate}
\end{lem}

\begin{proof}
The proof is analogous to that of Lemma~\ref{lem:intersectionA}, but longer as we 
must now divide into cases depending not only on 
$\sgn_l(\phi)$ but also on whether $l>a$, $l=a$ or $l<a$.
\end{proof}

\begin{lem}\label{lem:intersectionC} 
Let $\phi$ be a saturated formula over $f$ and $s\in(t_0,t_1)$. For every 
$\pe_{\bfI,\sigma}\in\Pe_{f,\bflambda}$ the following are equivalent:
\begin{enumerate}
\item[(2I1)] $\pe_{\bfI,\sigma}\cap \gv_{\delta_i,\varepsilon_i,\zeta,s}^{\syd,a}(f,\phi)\neq \varnothing$.
\item[(2I2)] $\sgn(\phi)\preccurlyeq \sigma$ and for all $l\in [q]$,
\begin{empheq}[left = \empheqlbrace]{align*}
|f_l|/\|f_l\| \leq\varepsilon_i\text{ on }\pe_{\bfI,\sigma},
&\text{\quad if }\sgn_l(\phi)=0,\\
|f_l|/\|f_l\| \geq\delta_i\text{ on }\pe_{\bfI,\sigma},
&\text{\quad if }\sgn_l(\phi)\neq 0\text{ and }l>a,\\
|f_l|/\|f_l\| > t_0\text{ on }\pe_{\bfI,\sigma},
&\text{\quad if }\sgn_l(\phi)\neq 0
\text{ and }l=a,\\
|f_l|/\|f_l\| \geq\zeta\text{ on }\pe_{\bfI,\sigma},
&\text{\quad if }\sgn_l(\phi)\neq 0\text{ and }l<a.
\end{empheq}
\end{enumerate}
Additionally, if any of the two claims above holds,
\begin{equation}\label{eq:intersectionC1}
\pe_{\bfI,\sigma}\cap \gv_{\delta_i,\varepsilon_i,\zeta,\tau}^{\syd,a}(f,\phi)=\begin{cases}
\pe_{\bfI,\sigma}\cap\{x\in\bbS^n\mid |f_a(x)|\geq \tau\},&
\text{if } |f_a|(t_0,t_1)\subseteq\pe_{\bfI,\sigma}, \\
\pe_{\bfI,\sigma},&\text{otherwise.}

\end{cases}
\end{equation}
\end{lem}

\begin{proof}
The implication from (2I1) to (2I2) is shown in a similar way as those  
from (0I1) to (0I3) in Lemma~\ref{lem:intersectionA} and from (1I1) to (1I3) 
in Lemma~\ref{lem:intersectionB}. We next prove the reverse implication. 

Assume then that (2I2) holds. From the conditions there and the definition of both 
$\pe_{\bfI,\sigma}$ and $\gv_{\delta_i,\varepsilon_i,\zeta,s}^{\syd,a}(f,\phi)$, it follows that
\begin{equation}\label{eq:comp1}
 \pe_{\bfI,\sigma}\cap \gv_{\delta_i,\varepsilon_i,\zeta,s}^{1,a}(f,\phi)
 =\pe_{\bfI,\sigma}\cap \{x\in\bbS^n\mid |f_a(x)|\geq s\}.
\end{equation}

We next divide in cases depending on whether 
$\pe_{\bfI,\sigma}\subseteq |f_a|^{-1}(t_0,t_1)$ or not.
\begin{enumerate}
\item[$\nsubseteq$)]If $|f_a|(t_0,t_1)\not\subseteq\pe_{\bfI,\sigma}$, then
$|f_a|\geq t_1$ on $\pe_{\bfI,\sigma}$, by (2I2), since $t_1$ is the next value in $\lambda_a$.  This shows that 
\begin{equation}\label{eq:comp2}
   \pe_{\bfI,\sigma}\cap \{x\in\bbS^n\mid |f_a(x)|\geq s\}=\pe_{\bfI,\sigma}.
\end{equation}
As  $\pe_{\bfI,\sigma}$ is non-empty, (2I1) follows from~\eqref{eq:comp1} 
and~\eqref{eq:comp2}.

\item[$\subseteq$)]If, instead, $|f_a|(t_0,t_1) \subseteq \pe_{\bfI,\sigma}$ then,
by Theorem~\ref{thm:flambdapartition}(2), the map
\[
  \pe_{\bfI,\sigma}\rightarrow (t_0,t_1),\, x\mapsto |f_{a}(x)|\|
\]
is surjective. Hence $\pe_{\bfI,\sigma}\cap \{x\in\bbS^n\mid |f_a(x)|\geq \tau \}$ 
is non-empty and (2I1) also follows in this case.
\end{enumerate}
We have proved~\eqref{eq:intersectionC1} in passing.
\end{proof}

Now we finish the proof of Proposition~\ref{prop:strataunion} with the help of the above three lemmas.

\begin{proof}[Proof of Proposition~\ref{prop:strataunion}]
Part~(1) follows directly from Lemmas~\ref{lem:intersectionA} and~\ref{lem:intersectionB}  
since these lemmas guarantee that each set in the right-hand side 
of~\eqref{eq:uniongvsets} is a union of strata.

We now show part~(2). Consider the intersections of 
$\pe_{\bfI,\sigma}$ with the decomposition~\eqref{eq:uniongvsets} for 
$\gv_\tau$ and $\gv_{t_0}$.

If for some $j\neq i$ and $\xi\in \Xi$ we have  
$\pe_{\bfI,\sigma}\cap\gv_{\delta_j,\varepsilon_j}(f,\phi_\xi)\neq\varnothing$,
then this intersection equals $\pe_{\bfI,\sigma}$ by Lemma~\ref{lem:intersectionA} 
and all the claims of~(2) 
hold trivially since $\pe_{\bfI,\sigma}\cap\gv_{\delta_j,\varepsilon_j}(f,\phi_\xi)$ does not depend on the value of $\tau$. 

Assume instead that for all $j\neq i$ and $\xi\in \Xi$ we have  
$\pe_{\bfI,\sigma}\cap\gv_{\delta_j,\varepsilon_j}(f,\phi_\xi)=\varnothing$.
Then 
\[
   \pe_{\bfI,\sigma}\cap \gv_{\tau}
   =\bigcup_{\xi\in\Xi}
   \pe_{\bfI,\sigma}\cap\gv^{\syd,a}_{\delta_i,\varepsilon_i,\zeta,\tau}(f,\phi_\xi)
\] 
and
\[
   \pe_{\bfI,\sigma}\cap \gv_{t_0}
   =\bigcup_{\xi\in\Xi}
   \pe_{\bfI,\sigma}\cap\gv^{\syd,a}_{\delta_i,\varepsilon_i,\zeta,t_0}(f,\phi_\xi).
\]

By hypothesis on $\pe_{\bfI,\sigma}$, we have 
$\pe_{\bfI,\sigma}\cap\gv_{t_0}=\pe_{\bfI,\sigma}
\neq \varnothing$ 
which implies that there exists $\xi\in\Xi$ such that $\pe_{\bfI,\sigma}\cap\gv^{\syd,a}_{\delta_i,\varepsilon_i,\zeta,t_0}(f,\phi_\xi)\neq\varnothing$.
Lemma~\ref{lem:intersectionB} then ensures that the conditions 
in~(1I3) hold true. But these conditions are the same as those in 
Lemma~\ref{lem:intersectionC}(2l2) except for $l=a$, where the inequality is strict in the latter and lax in the proper.
This means that $\pe_{\bfI,\sigma}\cap\gv^{\syd,i,a}_{\bfdelta,\bfeps,\zeta,\tau}(f,\Phi)=\varnothing$  
if and only if $|f_a|=t_0$ on $\pe_{\bfI,\sigma}$.
Furthermore, this latter condition is independent of the particular value of $\tau$.  
If it holds for $\tau$, then it holds for $\tau'$ and viceversa. This proves the first claim of~(2).

Arguing as above, we have that 
$\pe_{\bfI,\sigma}\cap\gv^{\syd,i,a}_{\bfdelta,\bfeps,\zeta,s}(f,\Phi)=\pe_{\bfI,\sigma}$ 
if and only if $|f_a|\geq t_1$ on~$\pe_{\bfI,\sigma}$.  
As this does not depend on the value of $\tau$, we get the second claim of~(2).

The third claim of~(2) follows directly from the last statement of Lemma~\ref{lem:intersectionC}.
\end{proof}

\section{Sampling theory for Gabrielov-Vorobjov approximations}
\label{sec:gamma}

In this section we prove the remaining stepping stones we introduced in 
the overview (Section~\ref{sec:overview}) namely, 
Theorems~\ref{theo:homotopywitness},~\ref{theo:homologywitness}, and~\ref{thm:samp}. 
The core of these proofs relies on the fact that the $\gamma$ invariant of Smale 
(see~\S\ref{sec:NSW} below) at points in $\bbS^n$ does not change much when we replace the 
homogeneous polynomials in $f$ by the perturbations $\fbar$. Otherwise, these three 
results are extensions of similar results in \cite{BCTC1} 
and most of the arguments we use are exactly those we used there to 
prove the simpler versions. We will therefore be concise and, in most places in what follows, 
limit our exposition to the general lines, omitting the details. 

\subsection{Semialgebraic sets from \texorpdfstring{$(f,t)$}{(f,t)} and related results}
\label{se:sample-rel}


Recall from Definition~\ref{def:laxF} the notion of a lax formula
over~$(f,t)$ and of its associated spherical sets $\Ap(f,t,\Phi)$.
We define now algebraic neighborhoods of the sets $\Ap(f,t,\Phi)$ 
by relaxing the inequalities.

\begin{defi}\label{def:AN} 
The \emph{algebraic neighborhood $\Ap_r(f,t,\Phi)$ of $\Ap(f,t,\Phi)$ with tolerance $r>0$} 
is the spherical semialgebraic set defined by replacing the atoms $f_i\geq t_j\|f_i\|$ 
by $f_i\geq (t_j-r)\|f_i\|$ and the atoms $f_i\leq t_j\|f_i\|$ by $f_i\leq (t_j+r)\|f_i\|$. 
By using strict inequalities, we define the 
\emph{open algebraic neighborhood $\Ap_r^\circ(f,t,\Phi)$ of $\Ap(f,t,\Phi)$ 
with tolerance $r$}. 
\end{defi}
\begin{remark}
We use the term ``algebraic neighborhood'' instead of the arguably more correct term ``semialgebraic neighborhood'' for the sake of conciseness and for consistency with~\cite{BCTC1}.
\end{remark}

The following three quantitative results about algebraic neighborhoods generalize, respectively, 
Proposition~4.17 and Theorem~4.19 
in~\cite{BCL17}, and Theorem~2.7 in~\cite{BCTC1}. 
In contrast with these result, 
the {\em separation} $\delta(t):=\inf_{i\neq j}|t_i-t_j|$ of $t$ 
enters here as a new parameter. 
Note that the sequence $t$ defined~\eqref{eq:seq_t} 
has the separation 
\begin{equation}\label{eq:delta(t)}
\delta(t)= \big(15(2m+1)D^2\sfK^2 \big)^{-1}.
\end{equation}

Similarly to~\eqref{eq:U}, we denote by $\mcU_{\bbS}(X,r)$ 
the {\em open (spherical) $r$-neighborhood} of a subset~$X$ of~$\bbS^n$,  
which is defined with respect to angular distance. 
Clearly, $\mcU_{\bbS}(X,r)\subseteq  \mcU(X,r)$.

\begin{prop}\label{prop:exclusion}
Let $f\in\Hd[q]$, $t\in\bbR^e$ and $r>0$. Then, for every lax formula $\Phi$ over~$(f,t)$, 
we have 
\begin{equation*}\label{eq:Prop_c}
 \mcU_{\bbS}(\Ap(f,t,\Phi),r)\subseteq \Ap_{D^{1/2}r}^\circ(f,t,\Phi) .
\end{equation*}
\end{prop}

\begin{prop}\label{prop:algebraicvsmetricnhood}
Let $f\in\Hd[q]$ and $T,r>0$ be such that $13\,D^{3/2}\kappabar(f)^2 (r+T)<1$. 
Then, for all $t\in(-T,T)^e$ satisfying $\delta(t)>2r$ and every purely conjunctive 
lax formula $\phi$ over $(f,t)$, we have 
\begin{equation*}
 \Ap_r^\circ(f,t,\phi)\subseteq \mcU_{\bbS}(\Ap(f,t,\phi),3\,\kappabar(f)r).
\end{equation*} 
\end{prop}

\begin{theo}[Generalized Quantitative Durfee Theorem]\label{theo:durfee}
Let $f\in\Hd[q]$ and $T,r>0$ be such that $\sqrt{2}\,\kappabar(f)(r+T)<1$. Then, for all 
$t\in(-T,T)^e$ satisfying $\delta(t)>2r$ and every purely conjunctive lax 
formula $\phi$ over $(f,t)$, the inclusions in
\begin{equation*}
  \begin{tikzcd}
  \Ap(f,t,\phi)\arrow[r,hook]  \arrow[dr,hook]
  & \Ap_r^\circ(f,t,\phi)\arrow[d,hook]\\
  & \Ap_{r}(f,t,\phi)
\end{tikzcd}
\end{equation*}
are homotopy equivalences.
\end{theo}

These three results are proved in the same manner as their corresponding  
results in~\cite{BCL17} and~\cite{BCTC1}. We briefly describe how the differences in the 
statements occur.

\begin{enumerate}
\item 
\emph{The condition $\sqrt{2}\,\kappabar(f)r<1$ becomes 
$\sqrt{2}\,\kappabar(f)(r+T)<1$}. In the case $t=0$, this condition guaranteed that 
when $\|f^L(x)\|/\|f^L\|<1/\big(\sqrt{2}\,\kappabar(f)\big)$ 
for a suitably chosen index set $L\subseteq [q]$, one could 
apply~\cite[Proposition~3.6]{BCTC1} to deduce  
that $\diff_xf^L$ is surjective. For such~$L$, we had $\|f^L(x)\|/\|f^L\|\leq r$ 
since the point $x$ was lying in $\Ap_r(f,\phi)\setminus\Ap(f,\phi)$. 
In our current setting, we will have $\|f^L(x)\|/\|f^L\|\leq r+T$ and 
so we have to replace $r$ by $r+T$ in the bound in the hypotheses to make the 
argument work.

\item \emph{The addition of the condition $\delta(t)>2r$}. 
This condition guarantees that the set 
$\Ap_r(f,t, (f_i\geq t_j\|f_i\|)\wedge (f_i\leq t_{j'}\|f_i\|)\wedge \psi)$ is  
empty whenever 
$t_{j'}<t_j$. This phenomenon is 
the only obstruction to assume without loss of 
generality that the formula $\phi$ 
in the statements of Proposition~\ref{prop:algebraicvsmetricnhood} and 
Theorem~\ref{theo:durfee} is of the form 
\[\bigwedge_{i\in I}(f_i\leq\|f_i\|t_{\alpha(i)})\wedge \bigwedge_{j\in J}(f_j\geq\|f_j\|t_{\alpha(j)}) \wedge \bigwedge_{k\in K}\left(f_k\leq\|f_k\|t_{\beta(k)})\wedge (f_k\geq\|f_k\|t_{\gamma(k)})\right)\] for some disjoint sets $I,J,K\subseteq \{1,\ldots,q\}$, and maps 
$\alpha:I\cup J\rightarrow\{1,\ldots,e\}$, $\beta:K\rightarrow \{1,\ldots,e\}$ and $\gamma:K\rightarrow \{1,\ldots,e\}$ such that for $k\in K$, $t_{\beta(k)}\geq t_{\gamma(k)}$.
\end{enumerate}
A complete proof reworking out all the details can be found in~\cite[Ch.~2]{tonellicuetothesis}. 
We note, however, that the proofs in~\cite{tonellicuetothesis} don't follow exactly the lines of the proofs in~\cite{BCL17,BCTC1}, although they use the same underlying fundamental ideas.

\subsection{Proof of Theorem~\ref{theo:homotopywitness}}
\label{sec:NSW}

Towards the proof of this result we recall the 
definition of the gamma invariant 
defined by Smale~\cite{Smale86} (see 
also~\cite{Bez4,Condition}) in both its affine and 
projective versions.
 
Let $n\ge m$. 
For a map $G:\bbR^{n+1}\rightarrow\bbR^m$ and a point $x\in\bbS^n$, 
we let $\diffa_xG:\bbR^{n+1}\to\bbR^m$ be the derivative at $x$ of $G$ viewed as a map on the Euclidean space 
and $\diff_xG:\Tg_x\bbS^n\to\bbR^m$ denote the derivative at $x$ of $G$ as a map on the sphere. 
{\em Smale's (Euclidean) gamma} of $G:\bbR^{n+1}\rightarrow\bbR^m$ at $x\in\bbS^n$ 
is the number given by
\begin{equation}\label{def:gamma}
  \affgamma(G,x):=\begin{cases}\sup_{k\geq 2}\left\|\diffa_xG^\dagger \frac{1}{k!}     
  \diffa^k_xG\right\|,&\text{if }\diffa_xG^\dagger\text{ surjective},\\
  \infty,&\text{otherwise}.\end{cases}
\end{equation}
The {\em projective Smale's gamma} $\gamma(G,x)$ of $G$ at $x$ is defined similarly, 
with $\diffa_xG$ replaced by $\diff_xG$. 

We recall from \cite{BCTC1} the condition number of $f\in\Hd[m]$ at $x$, 
defined by 
$\mu(f,x) := \|f\| \, \| D_x f^\dagger \Delta\|$, 
where 
$\Delta:=\mathrm{diag}(\sqrt{d_1},\ldots,\sqrt{d_m})$.
(Note that this quantity was denoted $\mu_{\mathrm{proj}}(f,x)$ in \cite{BCL17}.)
We also recall the Higher Derivative Estimate,   
$\gamma(f,x) \le \frac12 D^{3/2} \mu(f,x)$, 
from~\cite[Theorem~16.1]{Condition}. 
(This reference only covers the case $m=n$, but the proof given there 
immediately extends to $n\ge m$).

With these definitions, one proves the 
following proposition.

\begin{prop}\label{prop:boundgamma}
Let $f\in\Hd[m]$, with $n\geq m$, and $x\in\bbS^n$. Define $f_\bbS:=(f,\sum_{i=0}^nX_i^2-1)$. 
Then
\begin{equation}
2\affgamma(f_\bbS,x)\leq D^{\frac32}\mu(f,x)+D^{\frac12}\mu(f,x)\frac{\|f(x)\|}{\|f\|}+1.
\end{equation}
\end{prop}
\begin{proof}
By direct computation,
\[
\diffa^k_xf_{\bbS}(u_1,\ldots,u_k)=\begin{cases}
\begin{pmatrix}
\diffa_xf(u_1)\\2\langle x,u_1\rangle\end{pmatrix},&\text{if }k=1,\\
&\\
\begin{pmatrix}\diffa_x^2f(u_1,u_2)\\2\langle u_1,u_2\rangle\end{pmatrix},
&\text{if }k=2,\\
&\\
\begin{pmatrix}\diffa_x^kf(u_1,\ldots,u_k)\\0\end{pmatrix},
&\text{if }k>2.
\end{cases}
\]
Using this equality for $k=1$ we deduce that 
$\ker \diffa_xf_{\bbS}=\Tg_x\bbS^n\cap \ker \diffa_xf$. 
Let $V$ be the orthogonal complement of $\ker \diffa_xf$ in $\Tg_x\bbS^n$. 
Then 
$\left(\ker \diffa_xf_{\bbS}\right)^\perp=V+\bbR x$ and, 
for all $\lambda\in\bbR$, 
\begin{equation}\label{eq:diff-gamma}
\diffa_xf_\bbS(v+\lambda x)=\begin{pmatrix}
\diff_xf(v)+\lambda \Delta^2 f(x)\\2\lambda
\end{pmatrix}
\end{equation}
where, we recall, $\Delta:=\mathrm{diag}(\sqrt{d_1},\ldots,\sqrt{d_m})$ and 
$\diffa_xf(x)=\Delta^2 f(x)$ follows from Euler's identity for homogeneous 
functions.

By explicitly inverting the map in~\eqref{eq:diff-gamma}, we obtain 
\[(\diffa_xf_\bbS)^\dagger
\begin{pmatrix}
w\\t
\end{pmatrix}=
\diff_xf^\dagger\left(w-\frac{t}{2}\Delta^2f(x)\right)+\frac{t}{2}x.
\]
Thus
\[
(\diffa_xf_\bbS)^\dagger \frac{\diffa_x^kf_\bbS}{k!}(u_1,\ldots,u_k)
=\begin{cases}
\diff_xf^\dagger\frac{\diffa_x^2f}{2}(u_1,u_2)-\frac{\langle u_1,u_2\rangle}{2} \diff_xf^\dagger \Delta^2 f(x)+\frac{\langle u_1,u_2\rangle}{2} x,
&\text{if }k=2,\\
&\\
\diff_xf^\dagger\frac{\diffa_x^kf}{k!}(u_1,\ldots,u_k),&\text{if }k>2.
\end{cases}
\]
Applying the triangular inequality in~\eqref{def:gamma}, we obtain
\[
\affgamma(f_\bbS,x)\leq \gamma(f,x)+\frac{1}{2}\|\diff_xf^\dagger \Delta^2f(x)\|+\frac{1}{2} ,
\]
which implies 
\[
   2\affgamma(f_\bbS,x)\leq D^{\frac32}\mu(f,x)+D^{\frac12}\mu(f,x)\frac{\|f(x)\|}{\|f\|}+1
\]
where the first term in the right-hand side follows from the 
Higher Derivative Estimate and the second from the relations 
\[
\|\diff_xf^\dagger \Delta^2f(x)\|\leq \|\diff_xf^\dagger
\Delta\|\|\Delta\|\|f(x)\| 
=\|f\|\|\diff_xf^\dagger \Delta\|D^{\frac12}\frac{\|f(x)\|}{\|f\|}
\]
and the definition of $\mu$. This finishes the proof.
\end{proof}

\begin{proof}[Proof of Theorem~\ref{theo:homotopywitness}]
We refer to~\cite{BCL17} for the definition of the reach $\tau(X)$
of a subset $X$ of Euclidean space, and its local reach $\tau(X,p)$ at $p\in X$.
By~\cite[Theorem~2.8]{BCL17} (a variant of the Niyogi-Smale-Weinberger Theorem~\cite{NiSmWe2008}), 
it is sufficient to prove that 
\begin{equation}\label{eq:desideratum}
 \frac{1}{48 D^{\frac32}\kappabar(f)} \ < \ \frac{1}{2}\tau(\Ap(f,t,\phi)) .
\end{equation} 

Without loss of generality, we can write $\phi$ as $\bigwedge_{i\in I}(f_i\leq \|f_i\|t_{\beta(i)})\wedge \bigwedge_{j\in J}(f_j\geq \|f_j\|t_{\beta'(j)})$ 
for some $I,J\subseteq [q]$ and some maps $\beta:I\rightarrow [e]$ and $\beta':J\rightarrow [e]$.

For $L\subseteq I$ with $|L| \le n$ 
and $\alpha:L\rightarrow [e]$ we consider the (Euclidean) zero-set 
$$
  W_{L,\alpha} := \bigcap_{i\in L}\Ap\big(f_i - \|f_i\|t_{\alpha(i)} =0 \big).
$$
of $(f^L-\delta^{L,\alpha},\sum_{i=0}^nX_i^2-1)$, 
where we have put $\delta^{L,\alpha}:=(\|f_i\|t_{\alpha(i)})_{i\in L} \in\bbR^L$. 
By~\cite[Corollary~2.6]{BCL17}, we have 
\[
  \tau(\Ap(f,t,\phi))\geq \min\left\{\tau(W_{L,\alpha}) \mid L\subseteq I,\,|L|\leq n,\,\alpha:L\rightarrow [e]\right\} ,
\]
since the boundary of both sets $\Ap(f_i - \|f_i\|t_{j} \ge 0)$ and $\Ap(f_i - \|f_i\|t_{j} \le 0)$ 
equals $\Ap(f_i - \|f_i\|t_{j} = 0)$. 
The latter follows by the implicit function theorem, which can be applied 
due to~\cite[Proposition~3.6]{BCTC1} and $\sqrt{2}\,\kappabar(f)T<1$. 
Moreover, by~\cite[Theorem~2.11]{BCL17}, we have 
(recall the definition of $\bar{\gamma}$ in~\eqref{def:gamma})
\begin{equation}\label{eq:ineq1}
  \tau(W_{L,\alpha},p) \ \geq\ 
 \frac{1}{14\affgamma((f^L-\delta^{L,\alpha},\sum_{i=0}^nX_i^2-1),p)} .
\end{equation}

On the one hand, we have that 
$$
  \frac{\|f^L(p)\|}{\|f^L\|} \ \leq\ \max_{i\in L} \frac{|f_i(p)|}{\|f_i\|} \ <\  T<\frac{1}{\sqrt{2}\,\kappabar(f)} ,
$$
where for the second inequality, we have used that $p\in W_L$. 

On the other hand, we have 
$\mu(f^L,p)\leq\sqrt{2}\, \kappa(f^L)$, by~\cite[Proposition~3.6]{BCTC1}. 
 
Putting these together, 
we deduce with Proposition~\ref{prop:boundgamma} that 
\begin{align}\label{eq:ineq2}
2\affgamma\Big(\big(f^L-\delta^{L,\alpha},\sum_{i=0}^nX_i^2-1\big),p\Big)&< D^{\frac32}\mu(f^L,p)+D^{\frac12}\mu(f^L,p)T+1\nonumber\\ 
&<\sqrt{2}\,D^{\frac32}\kappa(f^L)+D^{\frac12}+1\leq (\sqrt{2}+2)D^{\frac32} \kappa(f^L),
\end{align}
where we used $\kappa(f^L) \ge 1 $ for the last inequality.  
Combining~\eqref{eq:ineq1} and~\eqref{eq:ineq2} the claim~\eqref{eq:desideratum} follows
and the proof is complete. 
\end{proof}

\subsection{Proof of Theorem~\ref{theo:homologywitness}}
\label{sec:HWT}

Theorem~\ref{theo:homologywitness} follows from the following more general result.

\begin{theo}
\label{theo:homologywitnessgen} 
Let $f\in\Hd[q]$ and $T>0$ be such that $2\kappabar(f)T<1$, and $\varepsilon>0$. 
Let $t\in(-T,T)^e$. 
Moreover, for $i \in [q]$ and $j \in [e]$, 
let $\mcX_{i,j}^\leq,\mcX_{i,j}^\geq\subseteq\bbS^n$ be finite subsets
such that for all purely conjunctive lax formulas $\phi$ over $(f,t)$, we have 
$$
   3 d_H\left(\phi\left(\mcX_{i,j}^{\propto}\mid i\in[q],\,j\in[e],\,\propto\in\{\leq,\geq\}\right), 
   \Ap(f,t,\phi)\right) 
   \ < \ \varepsilon \ < \ \min\left\{\frac{1}{48 D^{3/2} \kappabar(f)},\frac{\delta(t)}{12D^{1/2}}\right\}.
$$ 
Then, for all lax formulas $\Psi$ over $(f,t)$, 
the set $\Ap(f,t,\Psi)$ and the simplicial complex 
$$
 \Psi\Big(\cech{\varepsilon}{\mcX_{i,j}^\propto}\mid i\in [q],\,j\in [e],\, \propto\in\{\leq,\geq\}\!\Big)
$$ 
have the same homology.
\end{theo}

\begin{proof}
The proof follows the same lines as that of the Homology Witness 
Theorem~\cite[Theorem~2.4]{BCTC1}, now relying on 
Theorem~\ref{theo:durfee}, 
Theorem~\ref{theo:homotopywitness}, 
and Proposition~\ref{prop:exclusion} 
instead on Theorem~2.7, Theorem~2.3, and inequality~(2.4) 
in~\cite{BCTC1}, respectively. 

The main difficulty in the proof is to find some $r$ that satisfies certain inequalities so that the three results 
above can be applied. Because of the new bounds in the extended results, we now need to find 
$r>0$ such that both $\sqrt{2}\kappabar(f)(r+T)<1$ and $r<\frac12\delta(t)$ hold 
(hypothesis of Theorem~\ref{theo:durfee}) and $6D^\frac12\varepsilon\leq r$ (to conveniently use the  
bounds in Proposition~\ref{prop:exclusion}). Such $r$~can be found, provided 
\[
    6D^{\frac12}\varepsilon \ <\  \min\left\{ \frac{1}{\sqrt{2}\,\kappabar(f)}-T, \frac12\delta(t) \right\}.
\]
The assumption on $\varepsilon$ implies that 
$$
6D^{\frac12} \varepsilon \ < \  \min\left\{ \frac{1}{8D\kappabar(f)}, \frac12 \delta(t) \right\}
$$
Now we note that 
\[
   \min\left\{\frac{1}{8D\kappabar(f)}, \frac12\delta(t) \right\}<
   \min\left\{\frac{1}{\sqrt{2}\,\kappabar(f)}-T,\frac12\delta(t) \right\} ,
\]
since 
\[ 
 T<\frac{8D-\sqrt{2}}{8\sqrt{2}D}\frac{1}{\kappabar(f)} ,
\]
which in turn is guaranteed by the assumption 
$2\kappabar(f)T < 1$ and $D\geq 1$.
\end{proof}

\begin{proof}[Proof of Theorem~\ref{theo:homologywitness}]
Let $e=4m$ and $t$ be given by~\eqref{eq:seq_t} and put $T:=(2m+1)\delta(t)$. 
By Proposition~\ref{prop:kappa-est} 
and~\eqref{eq:delta(t)}, we have 
\[
\sqrt{2}\,\kappabar(f)T=\sqrt{2}\,\kappabar(f)\frac{1}{15D^2\sfK^2}<1.
\]
Also, Proposition~\ref{prop:kappa-est} and \eqref{eq:delta(t)}
imply that
\[
  \frac{1}{48D^{3/2}\kappabar(f)}>
  \frac{1}{180(2m+1)D^{5/2}\sfK^2}=\frac{\delta(t)}{12D^{1/2}} .
\]
Hence, all the hypotheses of 
Theorem~\ref{theo:homologywitnessgen} 
are satisfied and the conclusion follows from 
applying this result to formulas of the form $\Phibar$ over $(f,t)$ constructed from strict formulas~$\Phi$ over $f$.
\end{proof}

\subsection{Proof of Theorem~\ref{thm:samp}}
\label{sec:samp}

Again, Theorem~\ref{thm:samp} is an 
immediate consequence of the following more general result.

\begin{theo}\label{theo:sampling} 
Let $f\in\Hd[q]$ and $T,r>0$ be such that $13D^2\kappabar(f)^2(r+T)<1$.
Assume $t\in(-T,T)^e$ satisfies $\delta(t)>2D^{\frac12} r$ 
and let $\Phi$ be a strict 
formula over $(f,t)$.
Then for every finite set $\mcG\subseteq \bbS^n$ such that $d_H(\mcG,\bbS^n)<r$, we have
\[ 
 d_H\Big(\Ap^\circ_{D^\frac12r}(f,t,\Phi)\cap \mcG,\Ap(f,t,\Phi)\Big)<3D^\frac12\kappabar(f)r.
\]
\end{theo}

\begin{proof}
The proof is the same as that of~\cite[Theorem~6.5]{BCTC1}, 
using now Propositions~\ref{prop:algebraicvsmetricnhood} and~\ref{prop:exclusion} 
in the place of, respectively, Proposition~2.6 and inequality~(2.4) in~\cite{BCTC1}.
\end{proof}

\begin{proof}[Proof of Theorem~\ref{thm:samp}]
Let $e=4m$, $t$ be given by~\eqref{eq:seq_t} and $T:=\left(2m+1-\frac{1}{108D\sfK}\right)\delta(t)$. 
Using that 
$\delta(t)= \big(15(2m+1)D^2\sfK^2 \big)^{-1}$, 
we get 
\begin{equation}\label{eq:ell-delta}
r_\ell\leq \frac{\rho}{108D\sfK}\delta(t)\leq \frac{1}{108D\sfK}\delta(t).
\end{equation}
We verify that the hypothesis 
of Theorem~\ref{theo:sampling} is satisfied for $r=r_\ell$. 
On the one hand, 
using~\eqref{eq:ell-delta} and the definition of $T$, 
we get 
\[
13D^2\kappabar(f)^2(r_\ell+T)\leq
13D^2\kappabar(f)^2(2m+1)\delta(t)
=\frac{13\kappabar(f)^2}{15 D^2\sfK^2} < 1,
\]
where 
the last inequality follows from 
Proposition~\ref{prop:kappa-est}.
On the other hand, using~\eqref{eq:ell-delta} again, we have 
\[
2D^\frac12r_\ell< \delta(t).
\]
Finally, using~\eqref{eq:ell-delta} one more time, 
the result follows from Theorem~\ref{theo:sampling}, because   
\[
 3\cdot 3D^{\frac12}\kappabar(f)r_l\leq \frac{9\kappabar(f)\rho}{108D^{1/2}\sfK}\delta(t)\leq \frac{\rho}{12D^{1/2}}\delta(t),
\]
where the last inequality follows from Proposition~\ref{prop:kappa-est}. 
\end{proof}
\section{Concluding remarks}\label{sec:hybrid}

Cylindrical Algebraic Decomposition (CAD)~\cite{Collins75,Wut76} as refined in~\cite{SchSh83} may be used to compute the homology 
of semialgebraic sets with a worst-case complexity bounded by $(qD)^{2^{\Oh(n)}}$. One can combine CAD with Algorithm {\sf Homology} (by running both of them ``in parallel") to obtain a numeric-symbolic algorithm, call it {\sf Hybrid}, which enjoys (under infinite precision) 
the virtues of both CAD and {\sf Homology}. It has a  
weak singly exponential 
cost (thus exponentially accelerating the cost of CAD) with a doubly exponential worst-case cost 
(thus overcoming the major shortcoming of {\sf Homology}, 
the fact that it does not solve ill-posed data and that it takes 
too long for data close to ill-posed). 
\medskip

\noindent{\bf (1)}
The complexity bounds for {\sf Homology} (or {\sf Hybrid}), under infinite precision assumption, have a different nature than those for CAD. Algorithm {\sf Homology} has an input-dependent complexity bound and this bound is probabilistic, while the bound for CAD is input-independent and deterministic. This means that in order to make a fair comparison it would be helpful to have an answer to the following question.

\begin{question}
Is the worst-case complexity bound $(qD)^{2^{\Oh(n)}}$ of CAD attained for almost all inputs? Or can it be improved for a random input as in Theorem~\ref{thm:main_result}?
\end{question}

Ignoring this issue, and assuming that no better bounds are possible, we can say that Algorithm~{\sf Homology} is faster than CAD with high probability. In the general case, {\sf Homology} is faster than CAD (with probability at least 
$1-(nq D)^{-n}$) because
\[
\size(\Phi)q^{\Oh(n)}(nD)^{\Oh(n^3)}\leq (qD)^{2^{\Oh(n)}}.
\]
In the particular case of families of inputs for which the degree $D\ge 2$ is bounded and the number $q$ of polynomials is 
moderate, we have $N=n^{\Oh{(1)}}$. In this case, 
{\sf Hybrid} is faster than CAD (with probability at least $1-2^{-\size(p,\Phi)}$) because, under the $N=n^{\Oh{(1)}}$ assumption,
\begin{equation*}\label{eq:hybrid}
2^{\mcO\big(\size(p,\Phi)^{1+\frac{2}{D}}\big)}\leq 2^{\mcO\big(N^2\big)}\leq (qD)^{2^{\Oh(n)}}.
\end{equation*}

\noindent{\bf (2)} 
The above discussion assumes infinite precision. Under the 
presence of finite precision, the behaviors of {\sf Homology} and CAD are radically different. We have already observed that {\sf Homology} is numerically stable (and pointed the reader to~\cite[Section~7]{CKS16} for technical details). 
We can add some explanation here. The fundamental question 
in our context is the following: what is the finest precision required to ensure that the output of the algorithm is correct? No matter the algorithm, ill-posed inputs require infinite precision. And clearly, ill-conditioned inputs (i.e., those with a large condition number) will require a very large precision. But what about inputs with a moderate 
condition? The difference between {\sf Homology} and CAD 
becomes critical here. 

In a nutshell, CAD ends up 
performing computations with polynomials of doubly exponential degree. Round-off errors (even if they only occur when reading the input data) in the computations 
with these polynomials 
accumulate badly. In contrast, {\sf Homology} performs 
an exponential number of computations, corresponding to the points in the grid, with polynomial-size objects. And these 
computations are performed independently of each other. This results in a very moderate accumulation of errors. Hence, while one can prove that the precision needed for a correct 
answer with {\sf Homology} is small, this is certainly not the case with CAD. A recent result pointing in this direction is given in~\cite{NoTo:15}.
\smallskip

\noindent{\bf (3)} 
The weak cost bounds in Theorem~\ref{thm:main_result} 
depend on a choice of probability measure on the input space. 
Our choice (standard Gaussian in $\Pd[q]$, or uniform 
in the sphere $\bbS(\Pd[q])$, 
w.r.t.~the Weyl inner product) is the most common one 
for problems involving polynomial systems. But other choices are possible. In~\cite{EPR:17} a probabilistic analysis 
of $\kappa(f)$ in the case where 
$f$ is square ($n$ polynomials in $n+1$ homogeneous variables) is done which is valid for a broad class 
of probability distributions. A weak cost analysis, now valid for all distributions in this class, of the 
algorithm for counting zeros of square systems 
in~\cite{CKMW1} follows. Such a result raises the following 
question.

\begin{question}
Can one develop a probabilistic analysis of $\kappaff(p)$ for more general distributions? For the class of distributions in~\cite{EPR:17}, this reduces to developing a probabilistic analysis of $\kappa(f)$ in the underdetermined case.
\end{question}

\noindent{\bf (4)} 
Last but not least, it would be interesting to investigate the bit complexity of {\sf Hybrid}. That is, to study whether weak cost bounds similar to those in 
Theorems~\ref{thm:main_result} 
hold as well when the input is a 
tuple of integer polynomials. A particular case, 
with a clear interest, is that of the uniform distribution
on the set $\{-M,-M+1,-M+2,\ldots,M\}$. A positive result for this case would turn {\sf Hybrid} into a symbolic
algorithm more efficient than CAD for tuples of integer polynomials. 

\bibliographystyle{plain}
{\small 
\bibliography{biblio}
}
\end{document}